%% file: main.tex
\documentclass{article}[11pt]
\usepackage{geometry}
\usepackage{amsmath}
\usepackage{amssymb}
\usepackage{amsthm}
\usepackage{xspace}
\usepackage{mathrsfs}
\usepackage{algorithm}
\usepackage[noend]{algpseudocode}
\newcounter{protocol}
\makeatletter
\newenvironment{protocol}[1][htb]{%
  \let\c@algorithm\c@protocol
  \renewcommand{\ALG@name}{Protocol}
  \begin{algorithm}[#1]%
  }{\end{algorithm}
}
\makeatother
\usepackage{xcolor}
\usepackage{bbm}
\usepackage{extarrows}
\usepackage{pgfplots}
\usepackage{mathtools}
\usepackage{natbib}
\usepackage{subcaption}
\usepackage{breakcites}
\usepackage{graphicx}
\usepackage{soul}
\usepackage[colorlinks=true, allcolors=blue]{hyperref}
\usepackage{cleveref}
\crefname{procedure}{proc.}{procs.}
\Crefname{procedure}{Procedure}{Procedures}
\crefname{figure}{Figure}{Figures}
\Crefname{figure}{Figure}{Figures}

\newcommand\blfootnote[1]{%
  \begingroup
  \renewcommand\thefootnote{}\footnote{#1}%
  \addtocounter{footnote}{-1}%
  \endgroup
}

\usepackage{enumerate}

\usepackage{thmtools, thm-restate}

\newtheorem{theorem}{Theorem}

\newtheorem{lemma}[theorem]{Lemma}
\newtheorem{claim}[theorem]{Claim}
\newtheorem{fact}[theorem]{Fact}

\def\withnotes{1}

\ifnum\withnotes=1
    \newcommand{\todo}[1]{\textcolor{blue}{TODO: #1}}
    \newcommand{\Chris}[1]{\textcolor{brown}{Chris: #1}}
    \newcommand{\chris}[1]{\textcolor{brown}{Chris: #1}}
    \newcommand{\theertha}[1]{\textcolor{red}{theertha: #1}}
    \newcommand{\majid}[1]{\textcolor{magenta}{Majid: #1}}
\else
  \newcommand{\todo}[1]{\xspace}
  \newcommand{\Chris}[1]{\xspace}
  \newcommand{\chris}[1]{\xspace}
  \newcommand{\theertha}[1]{\xspace}
  \newcommand{\majid}[1]{\xspace}
\fi

\newcommand{\CSD}[1]{Drafter-Invariant Speculative Decoding}
\newcommand{\DISD}[1]{Drafter-Invariant Speculative Decoding}
\newcommand{\eqdef}{\mathbin{\stackrel{\rm def}{=}}}

\newcommand{\customfontsize}{\fontsize{8}{10}\selectfont}

\usepackage{soul}

\DeclareMathOperator{\argmin}{argmin}
\DeclareMathOperator{\argmax}{argmax}
\DeclareMathOperator{\mexp}{Exp}
\newcommand{\tv}{\mathrm{TV}}

\definecolor{lightblue}{rgb}{0.25, 0.41, 0.88}
\definecolor{lightred}{rgb}{0.88, 0.1, 0.3}

\title{Coupling without Communication \\
and Drafter-Invariant Speculative Decoding}
\author{
Majid Daliri\\New York University\\ \texttt{daliri.majid@nyu.edu}
\and
Christopher Musco\\New York University\\ \texttt{cmusco@nyu.edu}
\and
Ananda Theertha Suresh\\Google Research, NY\\ \texttt{theertha@google.com}
}

\input{bib_macros}

\begin{document}

\maketitle
\thispagestyle{empty}


\begin{abstract}
Suppose Alice has a distribution $\mathcal{P}$ and Bob has a distribution $\mathcal{Q}$. Alice wants to draw a sample $a\sim \mathcal{P}$ and Bob a sample $b \sim \mathcal{Q}$ such that $a = b$ with as high of probability as possible. It is well-known that, by sampling from an optimal coupling between the distributions, Alice and Bob can achieve $\Pr[a = b] = 1 - D_{\tv}(\mathcal{P},\mathcal{Q})$, where  $D_{\tv}(\mathcal{P},\mathcal{Q})$ is the total variation distance between $\mathcal{P}$ and $\mathcal{Q}$.

What if Alice and Bob must solve this same problem \emph{without communicating at all?} Perhaps surprisingly, with access to public randomness, they can still achieve $\Pr[a = b] \geq \frac{1 - D_{\tv}(\mathcal{P},\mathcal{Q})}{1 + D_{\tv}(\mathcal{P},\mathcal{Q})} \geq 1-2D_{\tv}(\mathcal{P},\mathcal{Q})$ using a simple protocol based on the Weighted MinHash algorithm. This bound was shown to be optimal in the worst-case by Bavarian, Ghazi, Haramaty, Kamath, Rivest, and Sudan [ToC 2020]. 


In this work, we revisit the ``communication-free coupling'' problem. We provide a simpler proof of the optimality result from [Bavarian et al., 2020]. Moreover we show that, while the \emph{worst-case} success probability of Weighted MinHash cannot be improved, an equally simple protocol based on Gumbel sampling offers a Pareto improvement: for every pair of distributions $\mathcal{P}$ and $\mathcal{Q}$,  Gumbel sampling achieves an equal or higher value of $\Pr[a = b]$ than Weighted MinHash.

Importantly, this improvement translates to practice. We demonstrate an application of communication-free coupling to \emph{speculative decoding}, a recent method for accelerating autoregressive large language models [Leviathan, Kalman, Matias, ICML 2023]. 
We show that communication-free protocols can be used to contruct \emph{\CSD{}} schemes, which have the desirable property that their output is fixed given a fixed random seed, regardless of what drafter is used for speculation. In experiments on a language generation task, Gumbel sampling outperforms Weighted MinHash.
Code is available at \url{https://github.com/majid-daliri/DISD}.

Finally, we study the coupling problem in the setting where communication is \emph{bounded}, rather than completely eliminated. We describe a protocol that uses just $O(\log(n/\epsilon))$ bits of communication to achieve $\Pr[a = b] = 1 - D_{\tv}(\mathcal{P},\mathcal{Q}) - \epsilon$, i.e. to essentially match optimal coupling.
\end{abstract}.

\pagenumbering{arabic} 

\blfootnote{We thank Cheuk Ting Li for bringing several related prior works to our attention after we initially published this paper. \cite{angel2021pairwiseoptimalcouplingmultiple} prove the same result that we do on the Gumbel sampling method, showing that it offers a Pareto improvement over Weighted MinHash. Moreover, in addition to our work and the cited work by \cite{BavarianGhaziHaramaty:2020}, \cite{li2019pairwisemultimarginaloptimaltransport,li2021unified} also prove that $\frac{1 - D_{\tv}(\mathcal{P},\mathcal{Q})}{1 + D_{\tv}(\mathcal{P},\mathcal{Q})}$ is optimal in the worst-case for communication-free coupling.
}

\newpage
\section{Introduction}
\label{sec:intro}
The goal of this paper is to revist a ``communication-free coupling'' problem that arises repeatedly in computer science, and was studied explicitly by \cite{BavarianGhaziHaramaty:2020} under the name \emph{correlated sampling}.
Suppose we have two cooperating parties, Alice and Bob. Alice has a distribution $\mathcal{P}$ and Bob has a distribution $\mathcal{Q}$, which, for simplicity, we assume are discrete distributions over the set $\{1, \ldots, n\}$. The goal is for Alice to sample $a\in \{1, \ldots, n\}$ according to $\mathcal{P}$ and Bob to sample $b\in \{1, \ldots, n\}$ according to $\mathcal{Q}$, in such a way that we maximize the probability that $a=b$. We are interested in how well Alice and Bob can solve this problem \emph{without communicating with each other}. We do allow them access to a source of shared random numbers. 

If communication was allowed, it is well known that Alice and Bob can achieve:
\begin{align}\label{eq:optimum_communication}
    \Pr[a = b] = 1 - D_{\tv}(\mathcal{P},\mathcal{Q}),
\end{align}
where $D_{\tv}(\mathcal{P},\mathcal{Q})$ is the total variation distance between $\mathcal{P}$ and $\mathcal{Q}$ \citep[Section 4.3]{wu2020information}. 
To do so, they construct an optimal coupling between the distributions, i.e., a distribution $\mathcal{C}$ over $\{1, \ldots, n\}\times \{1, \ldots, n\}$ with marginals $\mathcal{P}$ and $\mathcal{Q}$ that maximizes $\Pr_{(a,b)\sim \mathcal{C}}[a=b]$. In fact, the total variation distance is defined as $ \max_{\text{couplings } \mathcal{C}} \Pr_{(a,b)\sim \mathcal{C}}[a\neq b]$, so \eqref{eq:optimum_communication} is optimal by definition, with or without communication. 

While an optimal coupling between $\mathcal{P}$ and $\mathcal{Q}$ has a simple closed form (see \Cref{sec:prelims} for details), sampling from the coupling requires knowledge of \emph{both} $\mathcal{P}$ and $\mathcal{Q}$. I.e., sampling from an optimal coupling requires communication between Alice and Bob. However, it turns out that it is possible to sample from a \emph{near-optimal} coupling without any communication at all. In particular, the starting point of our work is a straightforward observation from \cite{BavarianGhaziHaramaty:2020}, which is implicit in prior work \citep{ManasseMcSherryTalwar:2010}:  
\begin{fact}
\label{fact:protocol}
There is a communication-free protocol (Protocol \ref{prot:minhash}) by which Alice samples $a\sim \mathcal{P}$ and Bob $b\sim \mathcal{Q}$ such that:\vspace{-1em}
\begin{align}
\label{eq:optimum_no_communication}
\Pr[a = b] \geq \frac{1 - D_{\tv}(\mathcal{P},\mathcal{Q})}{1 + D_{\tv}(\mathcal{P},\mathcal{Q})}.
\end{align}
\end{fact}
Since $0 \leq D_{\tv}(\mathcal{P},\mathcal{Q}) \leq 1$, \eqref{eq:optimum_no_communication} is always larger than $1 - 2D_{\tv}(\mathcal{P},\mathcal{Q})$. We conclude that it is possible to almost match the optimum in \eqref{eq:optimum_communication} without communication. \Cref{fact:protocol} follows from an analysis of the popular ``Weighted MinHash'' method \citep{KleinbergTardos:2002, Holenstein:2009,ManasseMcSherryTalwar:2010,Ioffe:2010,Shrivastava:2016,Christiani:2020}. We provide a sharp analysis of the method (with an exact expression for $\Pr[a = b]$) in \Cref{sec:protocols}, \Cref{clm:minhash}.

\subsection{Motivation}
The main question we address in this work is if it is possible to improve even further on the excellent performance of the existing Weighted MinHash method for communication-free coupling. To cut to the point: it is. We show that a simple method based on the ``Gumbel Max-Trick'' provides a Pareto improvement over Weighted MinHash. 
Before discussing this result in detail, however, we provide some motivation for studying communication-free coupling in such depth. It turns out that this problem has a direct application to a technique called \emph{speculative decoding} that has gained traction for its ability to accelerate inference in autoregressive language models \citep{LeviathanKalmanMatias:2023,ChenBorgeaudIrving:2023, SunSureshRo:2023}, one of the most important algorithmic challenges in modern machine learning \citep{HoffmannBorgeaudMensch:2022, zhou2024survey}. 

Details of speculative decoding are given in \Cref{sec:experiments}. In short, the method speeds up language models by using a small and inexpensive neural network to draft (i.e., speculate) the tokens that will be generated by a larger state-of-the-art network. Tighter coupling between the token distributions of the two networks leads to higher chance of draft tokens matching the larger model tokens, which 
leads to accelerated inference.

However, speculative decoding suffers from one major and subtle drawback: if the small ``drafter'' neural network changes (i.e., because of an update to improve the model), the use of an optimal coupling means that tokens generated by the state-of-the-art neural network will change. This can be potentially problematic in applications of autoregressive language models, where the expected and desired behavior is that the output of a model is \emph{fixed given a fixed random seed.} 
This behavior allows for researchers and other users to reliably reproduce results and leads to easier unit testing and debugging. 

As we will demonstrate in \Cref{sec:experiments}, communication-free coupling allows for the implementation of a \emph{\CSD} method that avoids this issue. In particular, the output of the state-of-the-art model is completely independent of what drafter network is used to accelerate inference: it is fixed given a fixed random seed. Intuitively this follows from the fact that in any communication-free coupling protocol, Alice and Bob's samples are necessarily independent when conditioned on their shared randomness.

\subsection{Our Contributions}
\label{sec:contributions}
Motivated by this application, our work seeks to understand if it is possible to improve on the existing Weighted MinHash method for communication-free coupling. Increasing $\Pr[a = b]$ from $\frac{1 - D_{\tv}(\mathcal{P},\mathcal{Q})}{1 + D_{\tv}(\mathcal{P},\mathcal{Q})}$ closer to the optimal $1- D_{\tv}(\mathcal{P},\mathcal{Q})$ could mean up to a 2x reduction in wrong guesses during speculative decoding, which would translate to a potential 2x reduction in computational depth during language generation. 

Unfortunately, improving on the bound in \Cref{fact:protocol} is not possible in the worst-case. As shown in an elegant proof by \cite{BavarianGhaziHaramaty:2020}, no protocol can do better than $\frac{1 - D_{\tv}(\mathcal{P},\mathcal{Q})}{1 + D_{\tv}(\mathcal{P},\mathcal{Q})}$ in the worst-case. We restate a version\footnote{The result in \cite{BavarianGhaziHaramaty:2020} is slightly stronger than our \Cref{clm:ub} in that it allows for any total variation distance, not just distances of the form $1/d$ for integer $d$. We assume distance of this form to simplify our argument.} of this result below and provide an even more compact proof in \Cref{sec:upper_bound}.

\begin{restatable}[]{theorem}{clmub}
\label{clm:ub}
Consider any protocol that takes as input a distribution and source of public randomness. For any positive integer $d$, there are distributions $\mathcal{P}$ and $\mathcal{Q}$ with total variation distance $D_{\tv}(\mathcal{P}, \mathcal{Q}) = 1/d$ such that, if Alice runs the protocol to sample $a\sim \mathcal{P}$ and Bob runs the protocol to sample $b\sim \mathcal{Q}$, then:
\begin{align*}
\Pr[a = b] \leq \frac{1 - 1/d}{1 + 1/d}.
\end{align*}
\end{restatable}
Importantly, while \Cref{clm:ub} rules out an improvement on Weighted MinHash for \emph{all}  distribution pairs, it does not mean that the method is optimal. Our main contribution is to demonstrate that another extremely simple protocol offers a \emph{Pareto improvement} over Weighted MinHash: for {any} $\mathcal{P},\mathcal{Q}$ it performs at least as well, but sometimes performs better (often significantly so in practice). 
Our alternative approach is  based on the so-called ``Gumbel Max-Trick'' or ``Gumbel sampling''. The method can be described in one sentence: using shared random variables $u_1, \ldots, u_n$ drawn uniformly from $[0,1]$, Alice returns $a = \argmin_{i\in \{1, \ldots, n\}} \frac{-\ln(u_i)}{p_i}$ \footnote{Often $ \argmin_{i\in \{1, \ldots, n\}} \frac{-\ln(u_i)}{p_i}$ is computed as $\argmax_{i\in \{1, \ldots, n\}} \ln p_i - \ln(-\ln(u_i))$. If $u$ is a uniform random variable in $[0, 1]$, then $-\ln(-\ln(u))$ follows a Gumbel distribution \citep{gumbel1935valeurs} and hence the name.} and Bob returns $b = \argmin_{i\in \{1, \ldots, n\}} \frac{-\ln(u_i)}{q_i}$, where $p_i$ and $q_i$ denote the respective probabilities that Alice and Bob's distributions, $\mathcal{P}$ and $\mathcal{Q}$, assign to item $i$. Intuitively, $a$ and $b$ will be partially coupled because, if some $u_i$ happens to be small, we are more likely to set both $a = i$ and $b = i$. 

Gumbel sampling has been widely studied in theoretical computer science under different names, including PPSWOR and bottom-k sampling \citep{Cohen:1997,Rosen:1997,Cohen:2023}. Moreover, it is already widely used in machine learning for uncoupled sampling from discrete distributions, including for auto-regressive language generation~\citep{kool2019stochastic}. So, using the method for speculative decoding would require essentially no code changes \citep{MaddisonTarlowMinka:2014,HuijbenKoolPaulus:2023}. 
We prove the following bound on the performances of Gumbel sampling in \Cref{sec:protocols}:

\begin{restatable}[]{theorem}{clmgumbel}
For distributions $\mathcal{P}$ and $\mathcal{Q}$, let  $C_{\text{WMH}}(\mathcal{P},\mathcal{Q})$ be the probability that samples $a\sim \mathcal{P}$ and $b\sim \mathcal{Q}$ drawn using the Weighted MinHash method (Protocol \ref{prot:minhash}) satisfy $a = b$. 
The communication-free Gumbel coupling method (Protocol \ref{prot:gumbel}) samples $a\sim \mathcal{P}$ and $b\sim \mathcal{Q}$ such that:
\begin{align*}
    \Pr[a=b] = \sum_{\substack{j\in \{1, \ldots, n\}\\\min(p_j,q_j) > 0}} \frac{1}{\sum_{i=1}^n \max(p_i/p_j, q_i/q_j)} \geq C_{\text{WMH}}(\mathcal{P},\mathcal{Q}) \geq \frac{1 - D_{\tv}(\mathcal{P}, \mathcal{Q})}{1 + D_{\tv}(\mathcal{P}, \mathcal{Q})}.
\end{align*}
\label{clm:gumbel}
\vspace{-.5em}
\end{restatable}
It is not hard to find distributions for which the first inequality above is strict, i.e., for which Gumbel sampling strictly improves on Weighted MinHash. In fact, as discussed in \Cref{sec:protocols}, this will \emph{always} be the case under the mild condition that $p_i \neq q_i$ for at least three values of $i$.  
As an example, consider distributions $\mathcal{P} = \{{1}/{2}, {1}/{2}, 0\}$ and $\mathcal{Q} = \{{1}/{3}, {1}/{3}, {1}/{3}\}$ over $\{1,2,3\}$. The best possible collision probability for these distributions is $1-\mathcal{D}_{TV}(\mathcal{P}, \mathcal{Q}) = {2}/{3}$ and $\frac{1 - D_{\tv}(\mathcal{P}, \mathcal{Q})}{1 + D_{\tv}(\mathcal{P}, \mathcal{Q})} = 1/2$. In this case, Gumbel sampling actually obtains the optimal collision probability of ${2}/{3}$, while Weighted MinHash obtains probability $7/12$, closer to the worst-case bound for communication-free methods.
A similar pattern emerges in our experiments with real-word next-token distributions that arise in language generation \Cref{sec:experiments}. Gumbel sampling often nearly matches the optimal coupling, obtaining collision probability close to $1-\mathcal{D}_{TV}(\mathcal{P}, \mathcal{Q})$, while Weighted MinHash performs somewhat worse. 

It is interesting to ask if additional improvement is possible. Is there another protocol that offers a Pareto improvement over Gumbel sampling, or is the approach Pareto optimal? I.e., does improvement for a pair of distributions $(\mathcal{P},\mathcal{Q})$ necessarily require worse performance for another pair?

Considering the worst-case bound of \Cref{clm:ub}, it is also natural to ask what is possible if we \emph{restrict}, but do not entirely eliminate communication. How many bits of communication between Alice and Bob  are required to ensure that $\Pr[a=b] = 1-D_{\tv}(\mathcal{P},\mathcal{Q}) -\epsilon$, i.e., to nearly match what is possible with optimal coupling? A baseline is to simply discretize $\mathcal{P}$ or $\mathcal{Q}$, communicate the entire distribution, and compute a near optimal coupling. 
This takes $O(n\log(n/\epsilon))$ bits of communication (see \Cref{sec:lowcomm} for details). We show that it is possible to get by with much less. In \Cref{sec:lowcomm} we prove:

\begin{restatable}[]{theorem}{lowcomm}
\label{clm:lowcomm}
There is a protocol that, for any $\epsilon \in (0,1)$, requires a constant number of rounds and $O(\log(n/\epsilon))$ bits of communication, in expectation, between Alice and Bob to produce samples $a\sim \mathcal{P}$, $b\sim \mathcal{Q}$ such that:\vspace{-1em}
\begin{align*}
    \Pr\left[a = b\right] \geq 1 - D_{\tv}(\mathcal{P}, \mathcal{Q}) - \epsilon.
\end{align*}
\end{restatable}

While we are not aware of immediate applications of \Cref{clm:lowcomm}, it provides an interesting point of comparison to both the full-communcation and no-communcation settings.

Finally, beyond our theoretical results, we implement our Gumbel sampling method and evaluate its performance in the ``drafter-invariant speculative decoding'' task in \Cref{sec:experiments}. We show that the method  nearly matches standard speculative decoding in terms of speculation accuracy, but with the added benefit that the neural network output is completely invariant to the drafter network utilized.

\subsection{Related Work}
\label{sec:related_work}

Our work is closely related to results on weighted coordinated sampling methods, which generalize well-known techniques like MinHash and the $k$-minimum values sketch  for \emph{unweighted} coordinated sampling \citep{Broder:1997,BroderCharikarFrieze:1998,CohenKaplan:2007,BeyerHaasReinwald:2007,LiChurchHastie:2006}. The goal in weighted coordinated sampling is similar to ours: Alice and Bob hold vectors, $A$ and $B$, and seek to independently produce a subsample of entries from their vector so that 1) indices corresponding to larger entries in the vectors are sampled with higher probability and 2) Alice and Bob return many of the same indices. 
The aim is to use Alice and Bob's subsamples to estimate functions involving interactions between corresponding entries in $A$ and $B$, like the inner product $\langle A,B\rangle = \sum_{i=1}^n A_iB_i$ \citep{Li:2017,BessaDaliriFreire:2023,DaliriFreireMusco:2024}. To do so effectively, it is critical that we have access to pairs $(A_i,B_i)$ with the same index $i$.

Weighted coordinated sampling methods include Weighted MinHash (also referred to as ``consistent weighted sampling'') \citep{ManasseMcSherryTalwar:2010,Ioffe:2010,HaeuplerManasseTalwar:2014,WuLiChen:2020}, Gumbel sampling \citep{Cohen:2015,Cohen:2023}, threshold sampling \citep{Flajolet:1990,DuffieldLundThorup:2005}, priority sampling \citep{DuffieldLundThorup:2004,DaliriFreireMusco:2024b}, and more \citep{EstanNaughton:2006}. While all of these methods intuitively seek to generate samples from $A$ and $B$ that contain the same indices with high probability, that is typically not the final goal: the methods are analyzed for specific downstream applications. As such we are unaware of any prior work besides that of \cite{BavarianGhaziHaramaty:2020} which specifically obtains bounds for our ``communication-free coupling'' problem. 

We also note that several coordinated sampling methods cannot be applied in our context. For example, priority sampling does not ensure that entries from $A$ and $B$ are sampled with probability exactly proportional to specified probabilities (we require samples to truly be drawn from $\mathcal{P}$ and $\mathcal{Q}$), and threshold sampling does not return a fixed number of samples (we always need exactly one sample). Of the major methods, this leaves Weighted MinHash and Gumbel sampling, which are the two methods we address in this paper.

\section{Preliminaries}
\label{sec:prelims}
\noindent\textbf{Notation.}
Throughout the paper, we assume two parties, Alice and Bob, who hold discrete distributions $\mathcal{P}$ and $\mathcal{Q}$ over the set $\{1,\ldots, n\}$. We use $p_i$ to denote the probability that $x$ drawn from $\mathcal{P}$ equals $i$ and $q_i$ to denote the probability that $x$ drawn from $\mathcal{Q}$ equals $i$. We will sometimes write $\mathcal{P} = \{p_1, \ldots, p_n\}$ and $\mathcal{Q} = \{q_1, \ldots, q_n\}$ since the list of probabilities completely specifies the distributions.

\medskip\noindent\textbf{Total Variation Distance.}
The total variation distance, $D_{\tv}(\mathcal{P},\mathcal{Q})$, between discrete distributions equals:
\begin{align}
\label{eq:dtv_equiv}
D_{\tv}(\mathcal{P},\mathcal{Q}) = \frac{1}{2}\sum_{i=1}^n |p_i-q_i| = \sum_{i=1}^n \max(0,p_i-q_i) = \sum_{i=1}^n \max(0,q_i-p_i).
\end{align}
Throughout our proofs, we will use two elementary inequalities that follow from \eqref{eq:dtv_equiv}:
\begin{align}
1 - D_{\tv}(\mathcal{P},\mathcal{Q}) &= \sum_{i=1}^n p_i - \max(0,p_i-q_i) = \sum_{i=1}^n \min(p_i,q_i)\label{eq:sum_mins_dtv}\\
1 + D_{\tv}(\mathcal{P},\mathcal{Q}) &= \sum_{i=1}^n p_i + \max(0,q_i-p_i) = \sum_{i=1}^n \max(p_i,q_i)  \label{eq:sum_maxs_dtv}
\end{align}

\noindent
\textbf{Coupling with Communication.} As discussed in \Cref{sec:intro}, if they are allowed to communicate, Alice and Bob can easily sample from $a \sim \mathcal{P}$ and $b\sim \mathcal{Q}$ in such a way that $\Pr[a=b] = 1 - D_{\tv}(\mathcal{P},\mathcal{Q})$. Concretely, they can execute the following standard protocol:
\begin{protocol}
    \caption{Coupling with Communication \citep{LeviathanKalmanMatias:2023, ChenBorgeaudIrving:2023}}
        \textbf{Protocol for Alice} (who has probability vector $\mathcal{P} = [p_1, \ldots, p_n]$):
        \begin{algorithmic}[1]
            \State Sample $a \sim \mathcal{P}$. Communicate $a$ and the distribution $\mathcal{P}$ to Bob.
        \end{algorithmic}
         \textbf{Protocol for Bob}  (who has probability vector $\mathcal{Q} = [q_1, \ldots, q_n]$):
        \begin{algorithmic}[1]
            \State Await for $(a, \mathcal{P})$ from Alice. With probability $\min(1, q_a/p_a)$ set $b = a$. 
            \State Otherwise, sample $b \sim \mathcal{Q}'$, where $\mathcal{Q}' = \{q_1', \ldots, q_n'\}$ is a distribution with $q_i' = \frac{\max(0,q_i - p_i)}{\sum_{j=1}^n \max(0,q_j - p_j)}$.
        \end{algorithmic}
    \label{prot:coupling}
\end{protocol}

Clearly Protocol \ref{prot:coupling} ensures that $a\sim \mathcal{P}$. It is easily checked that it also ensures that $b\sim \mathcal{Q}$. 
In particular, for any $i$ where $q_i \leq p_i$, $\max(0,q_i-p_i) = 0$, so we can only set $b = i$ in Step 1 of the protocol. We do so with probability $\Pr[a = i]\cdot q_i/p_i = q_i$, as desired. If $q_i \geq p_i$, then with probability $p_i$, we set $b=i$ in Step 1 of the protocol (exactly when $a$ is set to $i$). 
There is also some chance we set $b = i$ in Step 2 of the protocol, which we execute with probability $\sum_{j=1}^n p_j\cdot (1 - \min(1,q_j/p_j)) = \sum_{j=1}^n \max(0, p_j- q_j) = \sum_{j=1}^n \max(0, q_j- p_j)$. So, the overall probability we set $b=i$ in the second step is $q_i'\cdot \sum_{j=1}^n \max(0, q_j- p_j) = q_i - p_i$. Thus, the total probability we set $b=i$ (in either step) is $q_i - p_i + p_i = q_i$, as desired.

Finally, we can see that $\Pr[a = b] = \sum_{j=1}^n p_j \cdot \min(1,q_j/p_j) = \sum_{j=1}^n \min(p_j,q_j) = 1 - D_{\tv}(\mathcal{P},\mathcal{Q})$ via \eqref{eq:sum_mins_dtv}.

\section{An Existing Communication-free Protocol}
\label{sec:protocols}
As discussed, a communication-free protocol based on the Weighted MinHash algorithm can nearly match the performance of Protocol \ref{prot:coupling} (which requires communication, since Bob needs to know Alice's distribution $\mathcal{P}$). A number of papers analyze Weighted MinHash, including \cite{BavarianGhaziHaramaty:2020}. We provide a self-contained analysis below for completeness, and to obtain a sharp bound (i.e., an exact expression for $\Pr[a = b]$). 

\subsection{Weighted MinHash}
Technically speaking, there are many different ways to implement Weighted MinHash, which might be more accurately described as a family of closely-related sampling method. See \citep{WuLiChen:2020} for an overview. 
We analyze a particularly simple implementation suggested by a number of authors \citep{KleinbergTardos:2002,Holenstein:2009,Shrivastava:2016}. It is detailed in Protocol \ref{prot:minhash} and illustrated in \Cref{fig:dart_setup}.

\begin{protocol}
    \caption{Weighted MinHash Coupling}
        Fix public random numbers $u_1, u_2, u_3 \ldots $ drawn uniformly from the interval $[0,n]$.\\
        \textbf{Protocol for Alice} (who has probability vector $\mathcal{P} = [p_1, \ldots, p_n]$):
        \begin{algorithmic}[1]
        \For{$k=1,2,\ldots$}
            \If{$u_k \in [j-1, j-1+p_j]$ for some $j\in \{1, \ldots, n\}$}
                 return $a = j$.
            \EndIf
        \EndFor
        \end{algorithmic}
         \textbf{Protocol for Bob}  (who has probability vector $\mathcal{Q} = [q_1, \ldots, q_n]$):
        \begin{algorithmic}[1]
        \For{$k=1,2,\ldots$}
            \If{$u_k \in [j-1, j-1+q_j]$ for some $j\in \{1, \ldots, n\}$}
                  return $b = j$.
            \EndIf
        \EndFor
        \end{algorithmic}
    \label{prot:minhash}
\end{protocol}

\begin{figure}[ht!]
    \centering
    \includegraphics[width=0.9\linewidth]{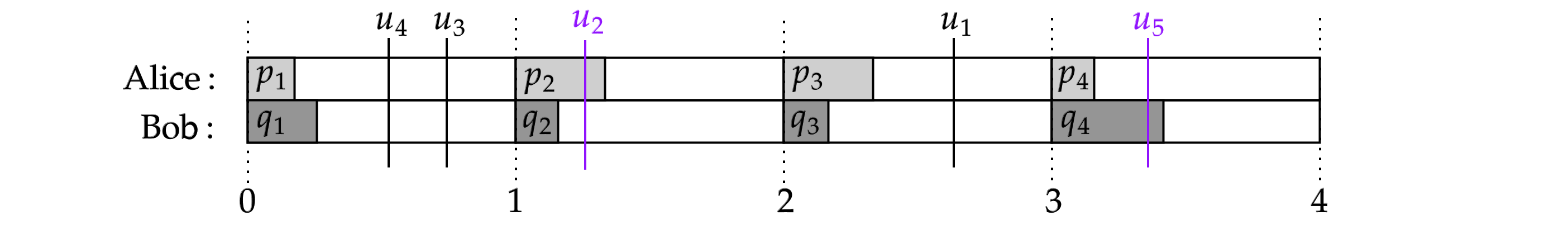}
    \caption{The Weighted MinHash method, described in Protocol \ref{prot:minhash}, selects a sample from a distributions on $n$ items with probabilities $p_1, \ldots, p_n$ by drawing a sequence of numbers $u_1, u_2, \ldots, $ uniformly at random from the interval $[0,n]$. The first time one of these numbers lands in a subinterval $[j-1, j-1+p_j]$ for any $j\in 1, \ldots, n$ , the protocol returns $j$. For example, in the illustration above, Alice returns $2$ (since $u_2$ falls within $[1, 1+p_2]$) and Bob returns $4$ (since $u_5$ falls within $[3, 3+p_4]$)}
    \label{fig:dart_setup}
\end{figure}

\begin{claim}
The communication-free Weighted MinHash coupling method (Protocol \ref{prot:minhash}) generates $a\sim \mathcal{P}$, $b\sim \mathcal{Q}$ such that:\vspace{-1em}
\begin{align*}
    \Pr[a=b] =\frac{1 - D_{\tv}(\mathcal{P}, \mathcal{Q}) + \sum_{i=1}^n |p_i-q_i| \min(p_i,q_i)}{1 + D_{\tv}(\mathcal{P}, \mathcal{Q})}  \geq \frac{1 - D_{\tv}(\mathcal{P}, \mathcal{Q})}{1 + D_{\tv}(\mathcal{P}, \mathcal{Q})}.
\end{align*}
\label{clm:minhash}
\end{claim}
\begin{proof}
As illustrated in \Cref{fig:dart_setup}, to run Protocol \ref{prot:minhash}, Alice draws values uniformly at random from the interval $[0,n]$. The protocol terminates and returns $j$ as soon as a value lands in any subinterval $[j-1,j-1 + p_j]$, which are shaded in the top row of \Cref{fig:dart_setup}. Since these subintervals have lengths $p_1, \ldots, p_n$, the probability that Alice returns $j$ is exactly $\frac{p_j}{\sum_{i=1}^n p_i} = p_j$. The proof that Bob returns $j$ with probability $q_j$ is identical. 

We turn our attention to the expression for $\Pr[a=b]$. Let $u_k$ denote be the first uniform random number for which either Alice or Bob terminates, i.e., the first $k$ for which $u_k \in [j-1,j-1 + \max(p_j,q_j)]$ for some $j$. We split the analysis into two cases: 
\begin{description}
\item[Case 1.]  Both Alice and Bob terminate at step $k$, i.e. $u_k \in [j-1,j-1 + \min(p_j,q_j)]$.
\item[Case 2.]  Only one of Alice and Bob terminate at step $k$, i.e. $u_k \in [j-1 + \min(p_j,q_j),j-1 + \max(p_j,q_j)]$.
\end{description}
In the first case, they always return the same index, $j = \lceil u_k \rceil$, so we have that $a =b$. Additionally, we can see that, conditioned on the event that $k$ is the first index for which $u_k \in [j-1,j-1 + \max(p_j,q_j)]$ for some $j$ (i.e., that $k$ is the first step that someone terminates at), the first case happens with probability:
\begin{align}
\frac{\sum_{i=1}^n\min(p_i,q_i)}{\sum_{i=1}^n\max (p_i,q_i)}. 
\end{align}
Next consider the second case. In this case, only one of Alice and Bob terminates at iteration $k$. The other terminates at iteration $k' > k$, and it may be that $a \neq b$.
To analyze the collision probability, we further break down the case, considering only when $u_k \in [j-1 + \min(p_j,q_j),j-1 + \max(p_j,q_j)]$ for a particular value of $j$.
We can see that, conditioned on the fact that one of the parties terminates at step $k$,
\begin{align*}
\Pr\left[u_k \in [j-1 + \min(p_j,q_j),j-1 + \max(p_j,q_j)]\right] = \frac{|p_j-q_j|}{\sum_{i=1}^n \max(p_i,q_i)}.
\end{align*}
If $u_k$ does land in the interval $[j-1 + \min(p_j,q_j),j-1 + \max(p_j,q_j)]$, Alice is the party who terminates and return $j$ whenever $p_j > q_j$. Otherwise, if $q_j > p_j$, Bob terminates and returns $j$. So, the question becomes, with what probability does the other party also returns $j$ on a future iteration? Future iterations involve random variables completely independent from $u_k$, so the probability is simply $\min(p_j,q_j)$. 

Putting together both cases, we have that:
\begin{align*}
    \Pr[a=b] = \frac{\sum_{i=1}^n\min (p_i,q_i)}{\sum_{i=1}^n\max (p_i,q_i)} \cdot 1 + \sum_{j=1}^n \frac{|p_j-q_j|}{\sum_{i=1}^n \max(p_i,q_i)} \cdot \min(p_j,q_j).
\end{align*}
The final result follows from rearranging and simplifying by applying \Cref{eq:sum_mins_dtv,eq:sum_maxs_dtv}.
\end{proof}

\subsection{Optimality of Communication-free Protocols}
\label{sec:upper_bound}
As discussed in \Cref{sec:intro}, when parameterizing performance by total variance distance, it can be proven that the bound of \Cref{clm:minhash} for Weighted MinHash is optimal for worst-case distributions. \cite{BavarianGhaziHaramaty:2020} do so via a reduction from an easier ``constrained agreement problem''. We provide an even simpler proof here, first restating the formal optimality claim from \Cref{sec:intro}:

\clmub*
\begin{proof}
As a warm up, we first consider a case involving distributions with total variation distance $1/2$, i.e., $d =2$. We will show that, without communication, it is not possible to achieve $\Pr[a=b] \geq \frac{1 - 1/2}{1+1/2} = 1/3$.
Specifically consider the following distributions $\mathcal{P},\mathcal{Q}$ and $\mathcal{R}$, which are each supported on three elements: 
\begin{align*}
\begin{bmatrix}
p_1 \\ p_2\\ p_3
\end{bmatrix} &= \begin{bmatrix}
1/2 \\ 1/2 \\ 0
\end{bmatrix} & 
\begin{bmatrix}
q_1 \\ q_2 \\ q_3
\end{bmatrix} &= \begin{bmatrix}
1/2 \\ 0 \\ 1/2
\end{bmatrix}& 
\begin{bmatrix}
r_1 \\ r_2 \\ r_3
\end{bmatrix} &= \begin{bmatrix}
0 \\ 1/2 \\ 1/2
\end{bmatrix}.
\end{align*}
Now, suppose a protocol is used to  sample $a\sim \mathcal{P},b\sim \mathcal{Q}$ and $c\sim \mathcal{R}$.
For any $i\in \{1,2,3\}$, it must be that one of $a,b$ or $c$ does not equal $i$, since at least one of the distributions has $0$ mass on item $i$. It follows that only \emph{one of} the equalities $a=b$, $a=c$, and $b=c$ holds at once. 
So, we have that for any protocol, $\min(\Pr[a=b], \Pr[a=c], \Pr[b=c]) \leq 1/3$. I.e., for two of the distributions, the protocol gives collision probability $\leq 1/3$. The result then follows from simply verifying that that $D_{\tv}(\mathcal{P},\mathcal{Q}) = D_{\tv}(\mathcal{P},\mathcal{R}) = D_{\tv}(\mathcal{Q},\mathcal{R}) = 1/2$:

We can generalize this result by having $d+1$ distributions $\mathcal{P}_1, \ldots, \mathcal{P}_{d+1}$ over $1, \ldots, d+1$. Specifically, let the $i^\text{th}$ distribution, $\mathcal{P}_i$, be uniform over all $j \neq i$. So, $x\sim \mathcal{P}_i$ has $1/d$ probability of equaling any $j \neq i$. It can be checked that $D_{\tv}(\mathcal{P}_i,\mathcal{P}_j) = 1/d$ for all $i,j$.

Now, consider a protocol that samples $x_1 \sim \mathcal{P}_1, \ldots, x_{d+1}\sim \mathcal{P}_{d+1}$. Consider all ${d+1 \choose 2}$ equalities of the form $\{x_i = x_j, i\neq j\}$. No matter what the outcome of $x_1, \ldots, x_{d+1}$, we  claim that at least $d$ of these inequalities will not be satisfied. In particular, as in the argument above, it cannot be that $x_1, \ldots, x_d$ all have the same value $v$ since one of $\mathcal{P}_1, \ldots, \mathcal{P}_{d+1}$ has zero mass for any possible $v$. As such, at least two items in $x_1, \ldots, x_{d+1}$ must be different from each other. 
Suppose without loss of generality that $x_1 \neq x_2$. 

First, the equality $x_1 = x_2$ is immediately not satisfied. Then for all $j \neq 1,2$, it must be that either $x_j \neq x_1$ or $x_j\neq x_2$, which is $d+1 - 2 = d-1$ additional inequalities not satisfied. 
So, no matter what the outcome of the protocol, the ratio of equalities of the form $\{x_i = x_j, i\neq j\}$ that are satisfied is at most: 
\begin{align*}
1 - \frac{d}{{d+1 \choose 2}} = 1 - \frac{2}{d+1}.
\end{align*}
It follows that there must be some pair $i,j$ such that:
\begin{align*}
\Pr[x_i = x_j] \leq 1 - \frac{2}{d+1} = \frac{1-1/d}{1+1/d}.&\qedhere
\end{align*}
\end{proof}

\section{Better Coupling via Gumbel Sampling}
Despite the limitation of \Cref{{clm:ub}}, in this section we show that it is nevertheless possible to strictly improve on the Weighted MinHash method from the previous section. To do so, we analyze a method based on Gumbel sampling, an existing protocol that is widely used for sampling from discrete distributions, even when coupling is not necessary. Under various names like PPSWOR, order sampling, and bottom-k sampling, the method has found widespread application across randomized algorithms  \citep{Rosen:1997,Cohen:2023,EfraimidisSpirakis:2006}. We give pseudocode in Protocol \ref{prot:gumbel}.

\begin{protocol}
    \caption{Gumbel Coupling}
        Fix public random numbers $u_1, \ldots, u_n$ drawn uniformly from the interval $[0,1]$.\\
        \textbf{Protocol for Alice} (who has probability vector $\mathcal{P} = [p_1, \ldots, p_n]$):
        \begin{algorithmic}[1]
            \State Return $a = \argmin_{i\in \{1, \ldots, n\}} \frac{-\ln(u_i)}{p_i}$.
        \end{algorithmic}
         \textbf{Protocol for Bob}  (who has probability vector $\mathcal{Q} = [q_1, \ldots, q_n]$):
        \begin{algorithmic}[1]
            \State Return $b = \argmin_{i\in \{1, \ldots, n\}} \frac{-\ln(u_i)}{q_i}$.
        \end{algorithmic}
    \label{prot:gumbel}
\end{protocol}

Our main result on Gumbel sampling is \Cref{clm:gumbel}, which we restate below:
\clmgumbel*

As discussed in \Cref{sec:intro}, it is not hard to find distributions for which the inequalities above are strict. In fact, Gumbel sampling often obtains a collision probability to the optimal $1 - D_{\tv}(\mathcal{P}, \mathcal{Q})$. Understanding this strong practical performance is an interesting direction for future work. 
It may be be that the bound of Theorem~\ref{clm:gumbel} tends to be much stronger for the sort of ``quickly decaying'' distributions common in applications. As an extreme example, suppose that $n=2$, or equivalently, $\mathcal{P}$ and $\mathcal{Q}$ only place mass on two items. Then, as above, Gumbel sampling has exactly the same collision probability as that of the optimal coupling.  To see this, without loss of generality, let $p_1 \leq q_1$, then the result of Theorem~\ref{clm:gumbel} simplifies as:
\begin{align*}
    \Pr[a=b] = \frac{1}{1 + \max\left(\frac{p_1}{p_2}, \frac{q_1}{q_2}\right)} + \frac{1}{1 + \max\left(\frac{p_2}{p_1}, \frac{q_2}{q_1}\right)} = \frac{1}{1 + \frac{q_1}{q_2}} + \frac{1}{1 + \frac{p_2}{p_1}} = q_2 + p_1 = 1 - D_{\tv}(\mathcal{P},\mathcal{Q}).
\end{align*}

We prove \Cref{clm:gumbel} in two parts, first showing that $\Pr[a = b] \geq \frac{1 - D_{\tv}(\mathcal{P}, \mathcal{Q})}{1 + D_{\tv}(\mathcal{P}, \mathcal{Q})}$. We then prove that $\Pr[a = b] \geq C_{\text{WMH}}$ which is more difficult. To do so, we require some standard preliminaries about exponential random variables. For a reference, see, e.g. \cite{Cohen:2023}.
\begin{fact}[Properties of Exponential Random Variables] 
\label{fact:exp_properties}
A random variable $X\sim \mexp(\lambda)$ is exponential with parameter $\lambda$ if $X = -\ln(u)/\lambda$ for a random variable $u$ drawn uniformly from $[0,1]$. Suppose we have independent exponential random variables $X_1\sim \mexp(\lambda_1)$ and $X_2\sim \mexp(\lambda_2)$. Then we have:
\begin{itemize}
    \item $\Pr[X_1 < X_2] = \frac{\lambda_1}{\lambda_1 + \lambda_2}$.
    \item $\min(X_1,X_2)$ is an exponential random variable with parameter $\lambda_1 + \lambda_2$.
\end{itemize}
\end{fact}
\subsection{Main Analysis}
\begin{proof}[Proof of \Cref{clm:gumbel}]
First note that:
\begin{align*}
    \Pr[a = b] = \sum_{j=1}^n \Pr[a = j \text{ and } b= j]. 
\end{align*}
So, we focus on analyzing $\Pr[a = j \text{ and } b = j]$ for a fixed $j$. Note that the probability is $0$ is either $p_j = 0$ or $q_j = 0$, so we assume below that $\min(p_j,q_j) > 0$. 
We have that $a = j$ and $b=j$ if for all $i \neq j$,
\begin{align*}
\frac{-\ln(u_j)}{p_j} &< \frac{-\ln(u_i)}{p_i} & &\text{and} & \frac{-\ln(u_j)}{q_j} &< \frac{-\ln(u_i)}{q_i}. 
\end{align*}
Equivalently, $a = j$ and $b=j$ if
\begin{align*}
-\ln(u_j) &< \min_{i\neq j} \left[-\ln(u_i)\cdot \frac{1}{\max(p_i/p_j, q_i/q_j)}\right].
\end{align*}
Observe that $-\ln(u_j)$ is an exponential random variable with parameter $\lambda = 1$ and, applying \Cref{fact:exp_properties},  $\min_{i\neq j} \left[-\ln(u_i)\cdot \frac{1}{\max(p_i/p_j, q_i/q_j)}\right]$ is exponential with parameter $\lambda = \sum_{i\neq j} \max(p_i/p_j, q_i/q_j)$. These two random variables are independent, so again by \Cref{fact:exp_properties}:
\begin{align}
\label{eq:exact_gumbel}
\Pr[a = j \text{ and } b = j] = \frac{1}{1 + \sum_{i\neq j} \max(p_i/p_j, q_i/q_j)}.
\end{align}
Summing over all $j$ for which $\min(p_j,q_j) > 0$ yields the exactly expression for $\Pr[a = b]$ in \cref{clm:gumbel}.

We next simplify this expression to prove that $\Pr[a=b] \geq \frac{1 - D_{\tv}(\mathcal{P},\mathcal{Q})}{1 + D_{\tv}(\mathcal{P},\mathcal{Q})}$, which establishes that Gumbel sampling at least matches the worst-case bound of \Cref{clm:ub}. Assume for now that $p_j \leq q_j$. We can make a symmetric argument for the case when $q_j > p_j$. If $p_j \leq q_j$, we have:
\begin{align*}
\Pr[a = j \text{ and } b = j] = \frac{p_j}{p_j + \sum_{i\neq j} \max(p_i, q_i\cdot \frac{p_j}{q_j})} = \frac{p_j}{\sum_{i=1}^n \max(p_i, q_i\cdot \frac{p_j}{q_j})}  \geq \frac{p_j}{\sum_{i=1}^n \max(p_i, q_i)}.
\end{align*}
To obtain the inequality, we use that $p_j/q_j \leq 1$ when $p_j \leq q_j$.
Similarly, when $q_j < p_j$, we have that $\Pr[a = j \text{ and } b = j] > \frac{q_j}{\sum_{i=1}^n \max(p_i, q_i)}.$
Summing over all $j$,
\begin{align*}
\Pr[a =b] \geq \frac{\sum_{j=1}^n \min(p_j,q_j)}{\sum_{j = 1}^n \max(p_i, q_i)} = \frac{1 - D_{\tv}(\mathcal{P},\mathcal{Q})}{1 + D_{\tv}(\mathcal{P},\mathcal{Q})}. 
\end{align*}

\subsubsection*{Pareto Improvement over Weighted MinHash}
Next, we strengthen the above claim by showing that, in fact, $\Pr[a=b] \geq C_{\text{WMH}}(\mathcal{P},\mathcal{Q})$. I.e., that Gumbel sampling offers a Pareto improvement over Weighted MinHash. 

To prove the stronger bound, we will analyze $\Pr[a = j \text{ and } b = j]$ separately for each $j \in 1, \ldots, n$. Again, we restrict to the case when $\min(p_j,q_j) > 0$, since if either probability equals $0$, then $\Pr[a = j \text{ and } b = j]$ equals $0$ for both Gumbel sampling and Weighted MinHash. When $a$ and $b$ are sampled from $\mathcal{P}$ and $\mathcal{Q}$ via Weighted MinHash, it can be seen from the analysis in \Cref{sec:protocols} that
$\Pr[a = j \text{ and } b = j] = {\min(p_j, q_j) \left(1 + |p_j - q_j|\right)}/{\sum_{i=1}^n \max(p_i, q_i)}.$
So, based on our exact expression for the collision probability for Gumbel sampling from \eqref{eq:exact_gumbel}, it suffices to prove that:
\begin{align}
\label{eq:to_prove_pareto}
\frac{\min(p_j, q_j) \left(1 + |p_j - q_j|\right)}{\sum_{i=1}^n \max(p_i, q_i)} 
    \leq 
    \frac{1}{\sum_{i=1}^n \max\left({p_i}/{p_j}, {q_i}/{q_j}\right)}.
\end{align}
Summing over all $j$ would then establish our desired bound. 

To prove \eqref{eq:to_prove_pareto}, we first note that, without loss of generality,  we can assume $p_j \leq q_j$ and let $q_j = k \cdot p_j$ for some $k \geq 1$. Substituting and rearranging, we can then restate the \eqref{eq:to_prove_pareto} as follows:
\begin{align*}
    p_j \left(1 + (k-1) p_j \right) \cdot \sum_{i=1}^n \max\left({p_i}/{p_j}, {q_i}/{k  p_j}\right) 
    &\leq  
    \sum_{i=1}^n \max(p_i, q_i).
\end{align*}
To prove the above, it suffices to show that:
\begin{align}
    \label{eq:main_to_prove}
    \sum_{i=1}^n S_i &\geq 0, & &\text{where} & S_i&\eqdef\left(\max(p_i, q_i) - \left(1 + (k-1) p_j \right) \cdot  \max\left(p_i, {q_i}/{k}\right) \right).
\end{align}
We analyze each term in the summation by splitting into two cases:

\begin{description}
    \item[Case 1: $p_i > \frac{q_i}{k}$.] In this case, we have:
    \begin{align*}
        S_i = \max(p_i, q_i) - \left(1 + (k-1) p_j \right) p_i
        \geq p_i - \left(1 + (k-1) p_j \right) p_i  = (1 - k) p_j p_i.
    \end{align*}
    \item[Case 2: $p_i \leq \frac{q_i}{k}$.] In this case, we have:
    \begin{align*}
    S_i = \left( 1 - \frac{1}{k} - (k-1) \frac{p_j}{k} \right) q_i 
     = \frac{k-1}{k} (1 - p_j) q_i \geq (k-1) (1 - p_j) p_i. 
    \end{align*}
\end{description}
We conclude that:
\begin{align*}
\sum_{i=1}^n S_i \geq 
(1 - k) p_j\cdot\left[ \sum_{i: p_i > \frac{q_i}{k}}  p_i\right] + (k-1)(1 - p_j)\cdot \left[\sum_{i: p_i \leq \frac{q_i}{k}}p_i\right].
\end{align*}
Finally, recall that, by definition, $q_j = kp_j$. So  the summation $\sum_{i: p_i > \frac{q_i}{k}} p_i$ excludes the case $i=j$, and thus must be $\leq 1-p_j$. Similarly, $\sum_{i: p_i \leq \frac{q_i}{k}} p_i \geq p_j$. Recalling that $k \geq 1$, so $(1-k)$ is non-positive and $(k-1)$ is non-negative, we conclude that:
\begin{align*}
    \sum_{i=1}^n S_i \geq 
    (1 - k) p_j(1-p_j) + (k-1)(1 - p_j)p_j = 0.
\end{align*}
We have thus proven equation \eqref{eq:main_to_prove} and \eqref{eq:to_prove_pareto}. As discussed, summing \eqref{eq:to_prove_pareto} over all $j$ with $\min(p_j,q_j) > 0$ establishes that, if $a,b$ are sampled from $\mathcal{P},\mathcal{Q}$ via Gumbel sampling, then $\Pr[a=b] \geq C_{\text{WMH}}(\mathcal{P},\mathcal{Q})$.
\end{proof}

\paragraph{Remark.}
We conclude by observing that the main inequality $\Pr[a=b] \geq C_{\text{WMH}}(\mathcal{P},\mathcal{Q})$ in \Cref{clm:gumbel} is actually strict under very mild conditions, i.e., in most cases Gumbel sampling strictly improves on Weighted MinHash.
In particular, consider the case when there are just three indices $j, k, l$ for which $p_j \neq q_j$, and $p_k \neq q_k$, and $p_l \neq q_l$, and one of the pairs $p_l,q_l$, $p_j,q_j$, or $p_k,q_k$ is non-zero. Then, it must be that, for two indices (without loss of generality, $j$ and $k$, either $p_j > q_j$ \emph{and} $p_k > q_k$, \emph{or} $q_j > p_j$ \emph{and} $q_k > p_k$. Without loss of generality, assume we are in the first case. Furthermore, without loss of generality, assume that $\frac{p_k}{q_k} \geq \frac{p_j}{q_j}$.

Now, consider our proof above where we analyze $\Pr[a = j \text{ and } b = j]$. We claim that at least one of the following inequalities from the end of the proof \emph{must} be strict:
\begin{align*}
    \sum_{i: p_i > \frac{q_i}{k}} p_i &\leq 1-p_j & &\text{or} & \sum_{i: p_i \leq \frac{q_i}{k}} p_i &\geq p_j.
\end{align*}
Specifically, if $\frac{p_j}{q_j} = \frac{p_k}{q_k}$, then $\sum_{i: p_i \leq \frac{q_i}{k}} p_i \geq p_j + p_k > p_k$, so the second inequality is strict. On the other hand, if $\frac{p_k}{q_k} > \frac{p_j}{q_j}$ then the first summation excludes \emph{both} $p_j, p_k$, which means $\sum_{i: p_i > \frac{q_i}{k}} p_i \leq 1 - p_j - p_k < 1 - p_j$.

Another example discussed in \cite{BavarianGhaziHaramaty:2020} is when $\mathcal{P}$ and $\mathcal{Q}$ are both uniform, but over different sized subsets (call them $\mathcal{A}$ and $\mathcal{B}$) of $\{1, \ldots, n\}$. In this special case, it can be checked that Gumbel sampling is exactly equivalent to Broder's ``unweighted MinHash'' method \citep{Broder:1997}. It yields collision probability $\frac{|\mathcal{A}\cap \mathcal{B}|}{|\mathcal{A}\cap \mathcal{B}|}$, whereas Weighted MinHash achieves a strictly worse probability of $\frac{\left(1+ \left|1/|\mathcal{A}| - 1/|\mathcal{B}|\right|\right)\cdot |\mathcal{A}\cap \mathcal{B}|}{|\mathcal{A}\cap \mathcal{B}| + |\mathcal{B}| - |\mathcal{A}|}$.

\section{Application: Drafter-Invariant Speculative
Decoding}
\label{sec:experiments}
With our main theoretical results in place, in this section we describe an application of communication-free coupling to accelerating autoregression language models, which have demonstrated impressive capabilities across a wide range of language-related tasks~\citep{touvron2023llama, achiam2023gpt, team2023gemini}. Given a text query $q$ (e.g., a question), the goal of autoregressive language models (LMs) is to generate a sequence of tokens $t_1,t_2,\ldots$, which correspond to words or word-pieces that comprise an answer or response to the query. At step $i$ of the generation, the  LM computes the conditional distribution $\text{Pr}(\cdot | q, t_1, \ldots, t_{i-1})$ and the next token is obtained by sampling according to this probability distribution. Computing this next-token distribution requires a forward pass through the neural network at every step $i$, which has significant computational costs, especially for large models.

Speculative decoding was proposed by \citet{LeviathanKalmanMatias:2023} and \citet{ChenBorgeaudIrving:2023} as a way to accelerate token generation. The method seeks to partially parallelize the process by using an inexpensive \emph{approximate distribution} (computed using a smaller neural network) to predict, or ``draft'', the next $\gamma$ tokens, $t_i, t_{i+1}, \ldots, t_{t+\gamma}$. Let the draft tokens be denoted by $\tilde{t}_i, \tilde{t}_{i+1}, \ldots, \tilde{t}_{t+\gamma}$. The larger neural network (or multiple copies of the network) can compute and sample from the distributions $\text{Pr}(\cdot | q, t_1, \ldots, {t}_{i-1})$, $\text{Pr}(\cdot | q, t_1, \ldots, {t}_{i-1},\tilde{t}_{i})$, $\text{Pr}(\cdot | q, t_1, \ldots, {t}_{i-1},\tilde{t}_{i},\tilde{t}_{i+1}), \ldots, \text{Pr}(\cdot | q, t_1, \ldots, {t}_{i-1},\tilde{t}_{i}, \ldots, \tilde{t}_{i+\gamma})$ \emph{in parallel}. If $\tilde{t}_i$ matches the token $t_i$ sampled from $\text{Pr}(\cdot | q, t_1, \ldots, {t}_{i-1})$, then the sample from $\text{Pr}(\cdot | q, t_1, \ldots, {t}_{i-1},\tilde{t}_{i})$ is a proper sample from the desired large model distribution. Further, if $\tilde{t}_{i+1}$ matches the token $t_{i+1}$ sampled from $\text{Pr}(\cdot | q, t_1, \ldots, {t}_{i-1},\tilde{t}_{i})$, then, likewise, the sample from $\text{Pr}(\cdot | q, t_1, \ldots, {t}_{i-1},\tilde{t}_{i},\tilde{t}_{i+1})$ is a valid sample from the large model. Overall, we can obtain $k$ samples in parallel if the first $k$ draft tokens match tokens sampled by the large model. We refer readers to \citet{LeviathanKalmanMatias:2023} for further details of the method.

In this setting, Bob corresponds to the large, expensive neural network, which knows the ``true'' distribution $\mathcal{Q}$ and Alice corresponds to the inexpensive neural network, which knows some distribution $\mathcal{P}$ that approximates $\mathcal{Q}$. We want Alice to sample a token from $\mathcal{P}$ that, with high probability, matches Bob's token sampled from $\mathcal{Q}$. Alice has to generate her sample before Bob, so the setting inherently does not allow communication from Bob to Alice. However, Alice can in principal communicate information to Bob that can be used when generating his sample. This is exactly what is done in current implementations of speculative decoding: Alice actually communicates the entire distribution she used to sample $\tilde{t}_{i+j}$, which allows Bob to sample in an optimally coupled way \citep[Algorithm 1]{LeviathanKalmanMatias:2023}. 

The end result is that Alice's predictions are accurate with probability exactly equal to $1 - D_{\tv}(\mathcal{P}, \mathcal{Q})$. However, this comes with a potential caveat: the output of the process depends on \emph{both} $\mathcal{P}$ and $\mathcal{Q}$. Hence, if the small ``drafter'' model changes, while the token distribution never changes, the exact token sampled by the large model very well could. This is undesirable: as discussed in \Cref{sec:intro}, the output of an LM is ideally {fixed given a fixed random seed}, no matter what optimizations are used to accelerate inference of the network. This property allows for users to easily reproduce results and facilitates easier unit testing and debugging. Naively, speculative sampling destroys this property. 

Fortunately, this issue can be fixed with a communication-free protocol! If we use a communication-free protocol to sample from the large model, then necessarily the models output is independent of the drafter model. We call the resulting approach ``Drafter-Invariant Speculative Decoding''. There is a price to pay for drafter-invariance: by using a communication-free protocol, we can only ensure that the draft tokens are correct with probability 
$c = \frac{1 - D_{\tv}(\mathcal{P},\mathcal{Q})}{1 + D_{\tv}(\mathcal{P},\mathcal{Q})}$ instead of 
$c = 1 - D_{\tv}(\mathcal{P},\mathcal{Q})$. Since the expected number of sequential tokens drafted correctly is $\frac{1}{1-c}$, this could lead to a roughly $2 \times$ reduction, although as we will see, the communication-free protocols analyzed in this work, especially our Gumbel sampling approach, tend to outperform the worst-case bound of \Cref{clm:ub} for real-world distributions arising from language models.

\begin{table}[ht]
  \centering
  \begin{tabular}{|c|c|c|c|c|c|}
    \hline
    \multicolumn{1}{|c|}{} & \multicolumn{4}{c|}{\textbf{Speculative Decoding}} \\
    \hline
    \multicolumn{1}{|c|}{} & \multicolumn{2}{c|}{\textbf{Drafter-Dependent (standard)}} & \multicolumn{2}{c|}{\textbf{Drafter-Invariant (ours)
    }} \\
    \hline
     No Drafter & Drafter: \small{Gemma 9B} & Drafter: \small{Gemma 2B} & Drafter: \small{Gemma 9B} &Drafter: \small{Gemma 2B}\\
    \hline
    \begin{tabular}[c]{@{}c@{}}\textbf{Yes}\\{}\end{tabular} &
    \begin{tabular}[c]{@{}c@{}}\textbf{Yes}\\{\color{lightred}\customfontsize Typically}\end{tabular} & 
    \begin{tabular}[c]{@{}c@{}}\textbf{Yes}\\{\color{lightred}\customfontsize **}\end{tabular} &  
    \begin{tabular}[c]{@{}c@{}}\textbf{Yes}\\{\color{lightred}\customfontsize It}\end{tabular} & 
    \begin{tabular}[c]{@{}c@{}}\textbf{Yes}\\{\color{lightred}\customfontsize No}\end{tabular} \\
    \hline
    \begin{tabular}[c]{@{}c@{}}\textbf{some}\\{}\end{tabular} &
    \begin{tabular}[c]{@{}c@{}}\textbf{often}\\{\color{lightblue}\customfontsize often}\end{tabular} & 
    \begin{tabular}[c]{@{}c@{}}\textbf{NP-hard}\\{\color{lightblue}\customfontsize NP-hard}\end{tabular} & 
    \begin{tabular}[c]{@{}c@{}}\textbf{some}\\{\color{lightred}\customfontsize but}\end{tabular} & 
    \begin{tabular}[c]{@{}c@{}}\textbf{some}\\{\color{lightred}\customfontsize NP}\end{tabular} \\
    \hline
    \begin{tabular}[c]{@{}c@{}}\textbf{NP-hard}\\{}\end{tabular} &  
    \begin{tabular}[c]{@{}c@{}}\textbf{within}\\{\color{lightred}\customfontsize with}\end{tabular} & 
    \begin{tabular}[c]{@{}c@{}}\textbf{problems}\\{\color{lightblue}\customfontsize problems}\end{tabular} & 
    \begin{tabular}[c]{@{}c@{}}\textbf{NP-hard}\\{\color{lightblue}\customfontsize NP-hard}\end{tabular} & 
    \begin{tabular}[c]{@{}c@{}}\textbf{NP-hard}\\{\color{lightblue}\customfontsize NP-hard}\end{tabular} \\
    \hline
    \begin{tabular}[c]{@{}c@{}}\textbf{problems}\\{}\end{tabular} & 
    \begin{tabular}[c]{@{}c@{}}\textbf{a}\\{\color{lightblue}\customfontsize a}\end{tabular} & 
    \begin{tabular}[c]{@{}c@{}}\textbf{often}\\{\color{lightred}\customfontsize can}\end{tabular} & 
    \begin{tabular}[c]{@{}c@{}}\textbf{problems}\\{\color{lightblue}\customfontsize problems}\end{tabular} & 
    \begin{tabular}[c]{@{}c@{}}\textbf{problems}\\{\color{lightblue}\customfontsize problems}\end{tabular} \\
    \hline
    \begin{tabular}[c]{@{}c@{}}\textbf{have}\\{}\end{tabular} &  
    \begin{tabular}[c]{@{}c@{}}\textbf{bounded}\\{\color{lightred}\customfontsize prov}\end{tabular} & 
    \begin{tabular}[c]{@{}c@{}}\textbf{have}\\{\color{lightblue}\customfontsize have}\end{tabular} & 
    \begin{tabular}[c]{@{}c@{}}\textbf{have}\\{\color{lightred}\customfontsize can}\end{tabular} & 
    \begin{tabular}[c]{@{}c@{}}\textbf{have}\\{\color{lightred}\customfontsize can}\end{tabular} \\
    \hline
    \begin{tabular}[c]{@{}c@{}}\textbf{efficient}\\{}\end{tabular} &
    \begin{tabular}[c]{@{}c@{}}\textbf{factor}\\{\color{lightblue}\customfontsize factor}\end{tabular} & 
    \begin{tabular}[c]{@{}c@{}}\textbf{efficient}\\{\color{lightblue}\customfontsize efficient}\end{tabular} & 
    \begin{tabular}[c]{@{}c@{}}\textbf{efficient}\\{\color{lightblue}\customfontsize efficient}\end{tabular} & 
    \begin{tabular}[c]{@{}c@{}}\textbf{efficient}\\{\color{lightred}\customfontsize fast}\end{tabular} \\
    \hline
    \begin{tabular}[c]{@{}c@{}}\textbf{approximation}\\{}\end{tabular} & 
    \begin{tabular}[c]{@{}c@{}}\textbf{of}\\{\color{lightblue}\customfontsize of}\end{tabular} & 
    \begin{tabular}[c]{@{}c@{}}\textbf{approximation}\\{\color{lightblue}\customfontsize approximation}\end{tabular} & 
    \begin{tabular}[c]{@{}c@{}}\textbf{approximation}\\{\color{lightblue}\customfontsize approximation}\end{tabular} & 
    \begin{tabular}[c]{@{}c@{}}\textbf{approximation}\\{\color{lightred}\customfontsize approximations}\end{tabular} \\
    \hline
    \begin{tabular}[c]{@{}c@{}}\textbf{algorithms.}\\{}\end{tabular} & 
    \begin{tabular}[c]{@{}c@{}}\textbf{the}\\{\color{lightblue}\customfontsize the}\end{tabular} & 
    \begin{tabular}[c]{@{}c@{}}\textbf{algorithms.}\\{\color{lightblue}\customfontsize algorithms.}\end{tabular} & 
    \begin{tabular}[c]{@{}c@{}}\textbf{algorithms.}\\{\color{lightblue}\customfontsize algorithms.}\end{tabular} &  
    \begin{tabular}[c]{@{}c@{}}\textbf{algorithms.}\\{\color{lightblue}\customfontsize algorithms.}\end{tabular} \\
    \hline
    \begin{tabular}[c]{@{}c@{}}\textbf{\texttt{<eos>}}\\{}\end{tabular} & 
    \begin{tabular}[c]{@{}c@{}}\textbf{optimal}\\{\color{lightblue}\customfontsize optimal}\end{tabular} & 
    \begin{tabular}[c]{@{}c@{}}\textbf{\texttt{<eos>}}\\{\color{lightblue}\customfontsize \texttt{<eos>}}\end{tabular} & 
    \begin{tabular}[c]{@{}c@{}}\textbf{\texttt{<eos>}}\\{\color{lightblue}\customfontsize \texttt{<eos>}}\end{tabular} &  
    \begin{tabular}[c]{@{}c@{}}\textbf{\texttt{<eos>}}\\{\color{lightblue}\customfontsize \texttt{<eos>}}\end{tabular} \\
    \hline
    - & 
    \begin{tabular}[c]{@{}c@{}}\textbf{solution.}\\{\color{lightblue}\customfontsize solution.}\end{tabular} & -
    & -
    &  
    - \\
    \hline
  \end{tabular}
\caption{
  We applied standard Speculative Decoding \citep{LeviathanKalmanMatias:2023} and our \CSD{} method to a 27 billion parameter Gemma model \citep{gemmateam2024gemma2improvingopen} to generate responses to the query: \emph{``Can NP-hard problems be approximated efficiently?''}. Smaller 9 and 2 billion parameter Gemma models were used as drafters.  
  The table show the tokens returned by the large model, as well as the draft tokens proposed by the smaller models (displayed below). If the draft token matches the large model, it is shown in {\color{lightblue} blue}; otherwise, it is show in {\color{lightred}red}. As we can see, \CSD{} always results in a response of \emph{``Yes some NP-hard problems have efficient approximation algorithms.''}, no matter what drafter is used. In contrast, standard Speculative Decoding leads to three different responses depending on the drafter. }
  \label{tab:sorting_algorithms}
\end{table}

\subsection{Experimental Evaluation}
We used Gumbel sampling to implement \CSD{} to accelerate generation of tokens for the 27 billion parameter \texttt{gemma-2-27b-it} model, using the smaller \texttt{gemma-2-9b-it}, and \texttt{gemma-2-2b-it} models as drafters \citep{gemmateam2024gemma2improvingopen}\footnote{\href{https://huggingface.co/docs/transformers/en/model\_doc/gemma2}{https://huggingface.co/docs/transformers/en/model\_doc/gemma2}}. 
Table \ref{tab:sorting_algorithms} shows that, for \CSD{}, the tokens generated by the large model remain the same regardless of the drafter model. In contrast, this is not the case for standard Speculative Decoding, which uses an optimal coupling.

\begin{figure}[h!]
    \centering
    \begin{subfigure}{.333\linewidth}
        \centering
        \includegraphics[width=\linewidth]{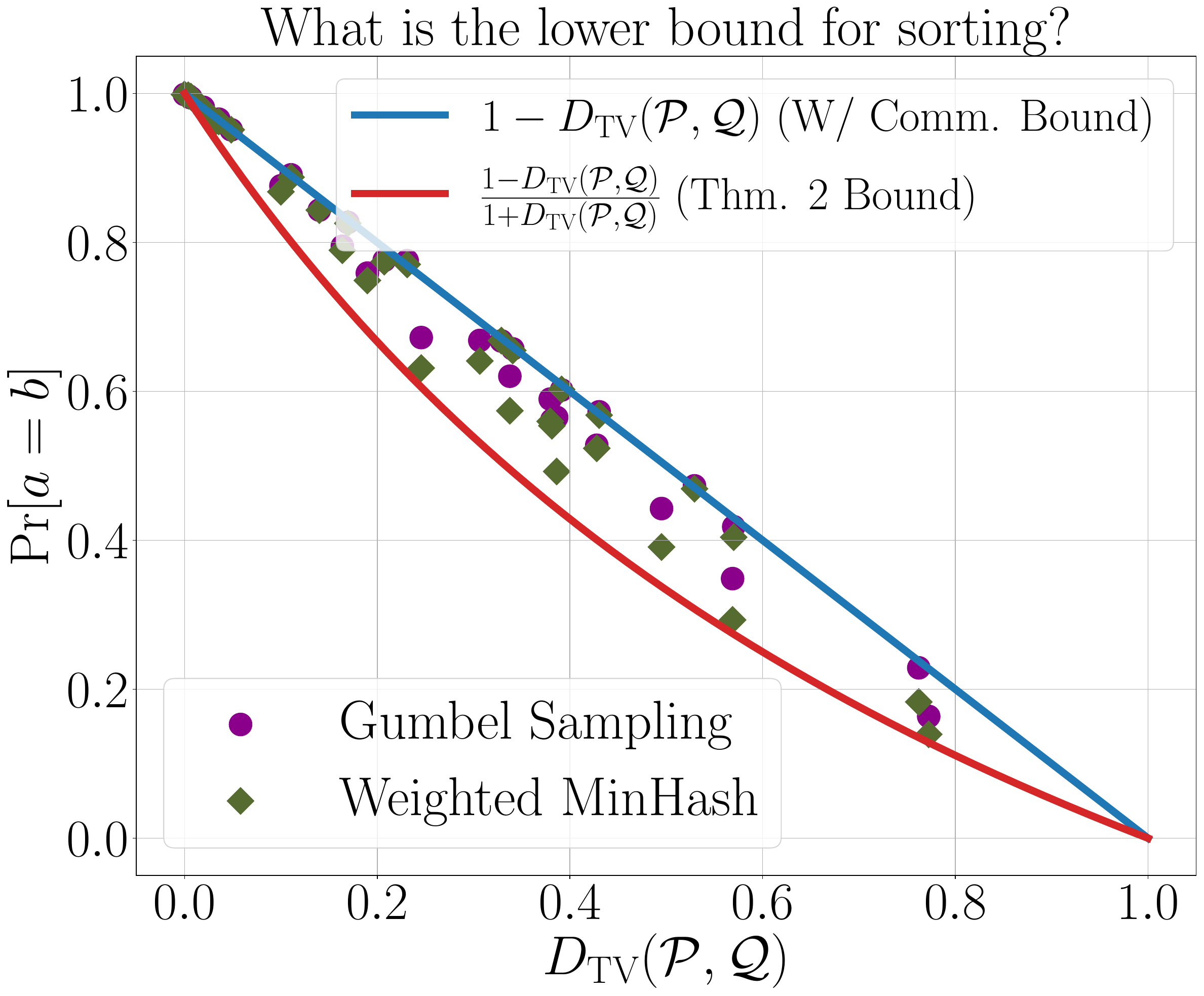}
    \end{subfigure}%
    \hfill
    \begin{subfigure}{.333\linewidth}
        \centering
        \includegraphics[width=\linewidth]{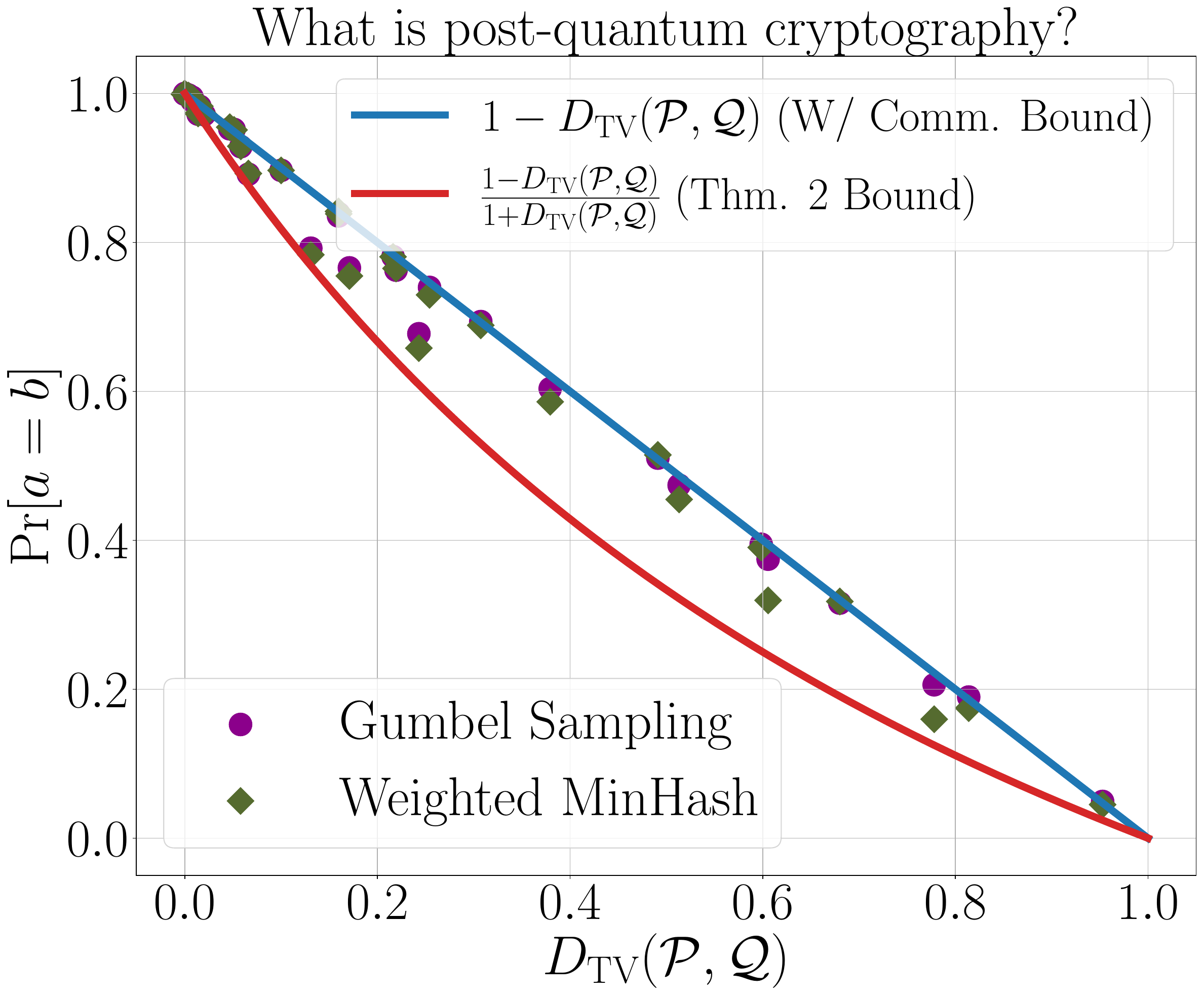}
    \end{subfigure}%
    \hfill
    \begin{subfigure}{.333\linewidth}
        \centering
    \includegraphics[width=\linewidth]{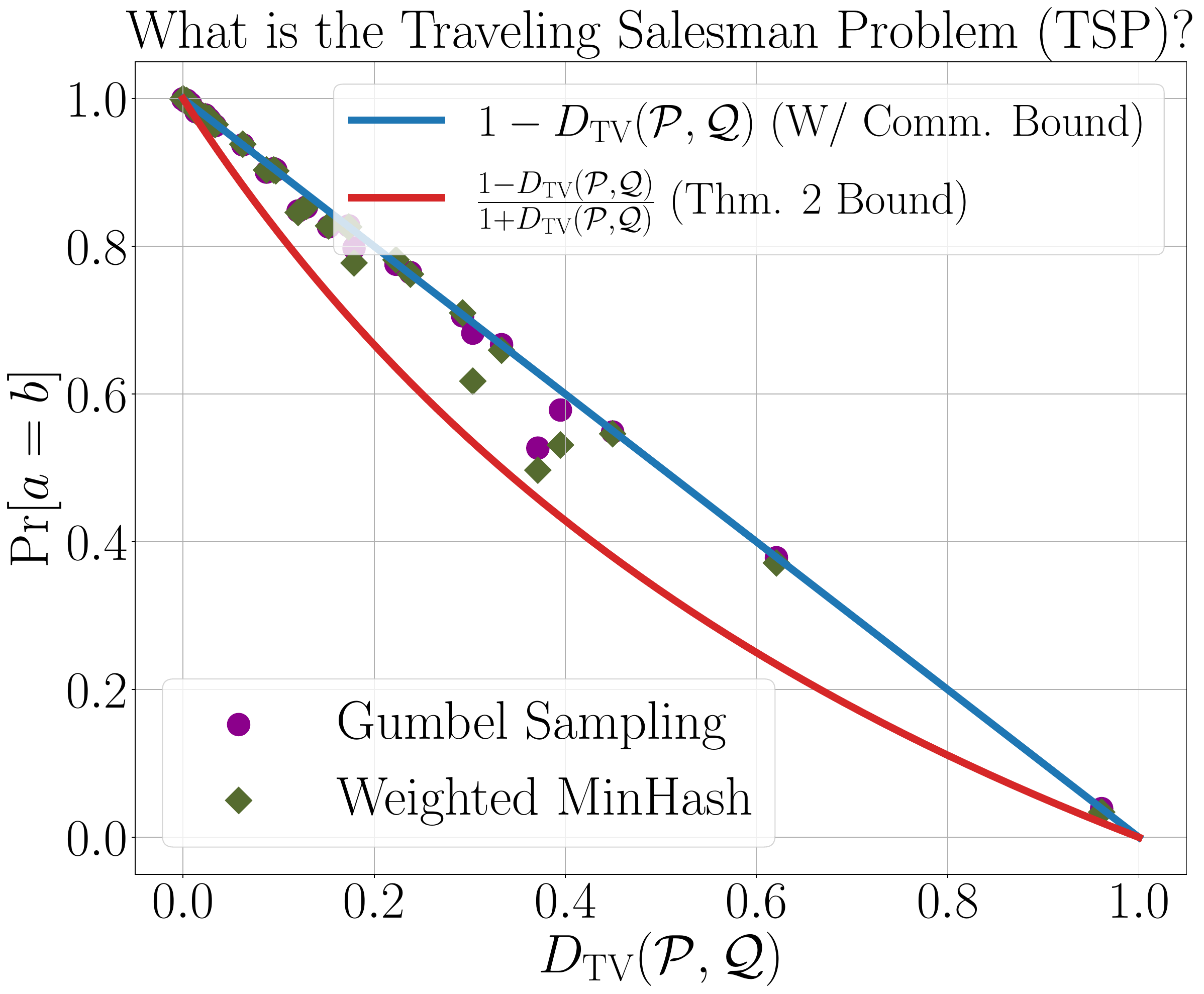}
    \end{subfigure}
    \caption{This plot illustrates the effectiveness of Gumbel and Weighted MinHash sampling for coupling samples. For each subplot, we obtain 32 pairs of distributions $(\mathcal{P},\mathcal{Q})$ by generating 32 response tokens for a given query using two different language models: the 27 billion parameter Gemma 27B model and the smaller Gemma 2B model. We sample $a\sim \mathcal{P}$ and $b\sim \mathcal{Q}$ 20,000 times using the no-communication protocols, and plot the empirical probability that $a=b$ vs. the total variation distance between $\mathcal{P}$ and $\mathcal{Q}$. For reference, we also plot the best possible probability that $a=b$, $1-D_{\tv}(\mathcal{P},\mathcal{Q})$, as well as our bound from \Cref{clm:ub}.}
    \label{fig:wmh_vs_gumbel}
\end{figure}

In \Cref{fig:wmh_vs_gumbel}, we also experimentally compare the effectiveness of Gumbel sampling and Weighted MinHash sampling for generating coordinated samples. We issue queries to the \texttt{gemma-2-27b-it} model, then generate a sequence of 32 response tokens using both the full model and the \texttt{gemma-2-2b-it} drafter model. For each token, we have a pair of distributions, $\mathcal{Q}$ (the large model's target distribution) and $\mathcal{P}$ (the drafter's approximate distribution). We generate samples $a\sim \mathcal{P}$ and $b\sim \mathcal{Q}$ from these distributions using the communication-free coupling protocols. We report the fraction of times that $a=b$ over $20,000$ repetitions using different random seeds. This results in 32 plotted points, per query, per method. Since the accuracy of the drafter model varies depending on where we are in the token generation sequence, we can see that the points corresponds to distributions with a fairly wide range of total variation distances. 

As we can see in the figure, Gumbel sampling achieves a higher collision probability than Weighted MinHash, across {all queries and tokens}, which is consistent with our \Cref{clm:gumbel}. Both methods typically outperform the worst-case bound of $\frac{1-D_{\tv}(\mathcal{P},\mathcal{Q})}{1+D_{\tv}(\mathcal{P},\mathcal{Q})}$ proven in \Cref{clm:ub}, often nearly matching the best possible collision probability of $1-D_{\tv}(\mathcal{P},\mathcal{Q})$.

\section{Coupling with Low-Communication}
\label{sec:lowcomm}
We conclude by considering a relaxed version of the communication-free coupling problem that allows for \emph{limited communication}. In particular, \Cref{clm:ub} proves that there is a gap between communication-free protocols and the communication-intensive protocol from Protocol \ref{prot:coupling}, where Alice sends her entire distribution to Bob. It is natural to ask what lies in between. Concretely, how many bits of communication are needed to match the optimal $1-D_{\tv}(\mathcal{P},\mathcal{Q})$ collision probability of Protocol \ref{prot:coupling}?

Our main result on this question is \Cref{clm:lowcomm}, which we restated below:
\lowcomm*

To prove \Cref{clm:lowcomm}, we first introduce an ``idealized'' communication protocol, Protocol \ref{prot:low_communication}, that proceeds in rounds. At each round, Alice and Bob communicate at most one index consisting of $O(\log n)$ bits, and one probability, i.e., a real-valued number in the interval $[0,1]$. We show that this protocol terminates with an expected constant number of rounds. We then show how to appropriately discretize Alice and Bob's distributions in a pre-processing step so that the probability can be communicated with $O(\log(n/\epsilon))$ bits.
The same discretization strategy can be used to give a concrete communication complexity bound for Protocol \ref{prot:coupling}, which naively requires communicating real-valued probabilities. The result would be a protocol requiring $O(n\log(n/\epsilon))$ bits of communication.
\Cref{clm:lowcomm} offers and exponential improvement on this baseline.

\begin{protocol}[ht]
    \caption{Low Communication Coupling}
    \label{prot:low_communication}
    Fix public random numbers $u_0, u_1, \ldots$ drawn uniformly from the interval $[0,1]$.\\
    \textbf{Protocol for Alice} (who has probability vector $\mathcal{P} = [p_1, \ldots, p_n]$):
    \begin{algorithmic}[1]
        \State Sample $a \sim \mathcal{P}$ and send $(a, p_a)$ to Bob.
        \State Await Bob's response.
            \If{Bob sends {\color{lightred} \textbf{reject}}}
                \For{$k=1,2,\ldots$}
                    \State Await $(j, w)$ from Bob.
                    \State If $p_j \geq u_k - w$, send {\color{lightred} \textbf{reject}} to Bob. Otherwise send {\color{lightblue} \bf approve} and break the loop.
                \EndFor
            \EndIf
            \State Return $a$.
    \end{algorithmic}
        \textbf{Protocol for Bob} (who has probability vector $\mathcal{Q} = [q_1, \ldots, q_n]$): 
    \begin{algorithmic}[1]
        \State Receive $(a, p_a)$ from Alice.
        \If{$u_0 \leq \min\left(\frac{q_a}{p_a}, 1\right)$}
            \State Send {\color{lightblue} \bf approve} to Alice.
            \State Return $b=a$.
        \Else
        \State Send {\color{lightred} \bf reject} to Alice.
        \For{$k=1,2,\ldots$}
            \State Find $j$ such that $\sum_{t=1}^{j-1}q_t \leq u_k  < \sum_{t=1}^{j} q_t$ and send $(j, \sum_{t=1}^{j-1} q_t)$ to Alice.
            \State Await Alice's response.
            \If{Alice sends {\color{lightblue} \bf approve}}
                \State Return $b=j$
            \EndIf
        \EndFor
        \EndIf
    \end{algorithmic}
\end{protocol}

\smallskip\noindent
\textbf{Idealized Protocol.} We begin by analyzing our idealized protocol,
Protocol \ref{prot:low_communication}. It is based on a modification of the optimal coupling method, Protocol \ref{prot:coupling}. While Protocol \ref{prot:coupling} naively requires communicating Alice's entire distribution, $\mathcal{P}$, to Bob, we can see that $\mathcal{P}$ only becomes necessary if Bob ``rejects'' the item $a$ that
he receives from Alice. In particular, upon rejectionm he must sample from a distribution $\mathcal{Q}'$ that depends on $\mathcal{P}$.

\begin{figure}[ht]
    \centering
    \includegraphics[width=0.8\linewidth]{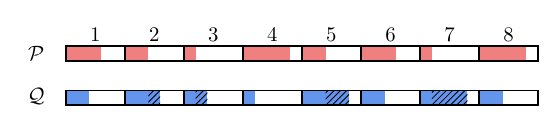}
    \caption{Two probability distributions, $\mathcal{P}$ (red) and $\mathcal{Q}$ (blue), are depicted in this diagram. The area where $q_i > p_i$ is highlighted with a hatch pattern. In the first round of Protocol \ref{prot:low_communication}, Alice samples an item proportional to $\mathcal{P}$. Using rejection sampling, Bob decides whether to accept or reject the item. If rejected, Bob attempts to draw a random number that ``hits'' the hatched regions where $q_i$ exceeds $p_i$.}
    \label{fig:distribution_difference_visualization}
\end{figure}

Protocol \ref{prot:low_communication} implements this ``post-rejection'' step in a more communication efficient way by using a dart-throwing approach similar to Weighted MinHash. Specifically, Bob attempts to sample from regions of the number line where  $q_j > p_j$, as illustrated by the hatched regions in \cref{fig:distribution_difference_visualization}. A uniform sample from these regions amounts to a sample from $\mathcal{Q}'$. To obtain such a sample, Bob throws a dart and, if it lands in the region $[j-1, j-1 + q_j]$ for some $j$, he sends the dart to Alice to verify that it also lies above  $j-1 + p_j$. If it does, then the dart has landed in the hatched region, so they terminate the process. If not, Bob tries again. We prove that the number of rounds required for the process to terminate is constant in expectation:

\begin{claim}\label{claim:exp_communication}
If Alice and Bob use Protocol \ref{prot:low_communication} to produce samples $a \sim \mathcal{P}$ and $b \sim \mathcal{Q}$, then:
\begin{align*}
    \Pr[a = b] = 1 - D_{\tv}(\mathcal{P}, \mathcal{Q}).
\end{align*}
Furthermore, the expected number of messages before the protocol terminates is at most $2$.
\end{claim}

\begin{proof}
Let $a$ and $b$ be the indices returned by Alice and Bob at the end of the protocol. We first prove that $a \sim \mathcal{P}$, $b \sim \mathcal{Q}$, and $\Pr[a = b] = 1 - D_{\tv}(\mathcal{P}, \mathcal{Q})$. The fact that $a\sim \mathcal{P}$ is immediate, since Alice starts by choosing $a\sim \mathcal{P}$ and then never changes $a$'s value. 
Alice's selection is communicated to Bob, who then decides to either accept or reject $a$. Bob accepts with probability $\min\left(\frac{q_a}{p_a}, 1\right)$. So, fixing an index $i$, we have:
\begin{align*}
\Pr[a = b = i] = p_i\cdot \min\left(\frac{q_i}{p_i}, 1\right) = \min(q_i, p_i).
\end{align*}
Summing over all $i$ we conclude that
\begin{align*}
\Pr[a = b] = \sum_{i=1}^{n} \Pr[a = b = i] = \sum_{i=1}^{n} \min(p_i, q_i) = 1 - D_{\tv}(\mathcal{P}, \mathcal{Q}).
\end{align*}
Now, if Bob rejects $a$ at Line 6, he must select another value $j \neq a$. The probability of selecting $j$ given the rejection of $a$ is:
\begin{align*}
\Pr[b = j \mid \text{Bob rejects $a$}] = \frac{q_j - \min(p_j, q_j)}{\sum_{k=1}^{n} q_k - \min(p_k, q_k)}. 
\end{align*}
The bottom of the fraction above is exactly the probability that Bob rejects at Line 6, so overall, $\Pr[b = j] = \Pr[b = j \land b = a] +  \Pr[b = j \land b \neq a] = \min(q_j,p_j) +  q_j -  \min(q_j,p_j)  = q_j$, as desired.

Next, we bound the round complexity of Protocol \ref{prot:low_communication}.
Bob's process involves a dart-throwing mechanism where each throw is aimed at hitting a region representative of the excess probability mass of $q_j$ over $p_j$. The chance of hitting the correct region per trial, shown as the hatched area in the \cref{fig:distribution_difference_visualization}, is equal to the total variation distance, $D_{\tv}(\mathcal{P}, \mathcal{Q})$. 

\begin{align*}
    \sum_{j=1}^{n} q_j - \min(p_j, q_j) = \sum_{j=1}^{n} q_j - \sum_{j=1}^{n} \min(p_j, q_j) = 1 - (1 - D_{\tv}(\mathcal{P}, \mathcal{Q})) = D_{\tv}(\mathcal{P}, \mathcal{Q}).
\end{align*}
Each throw is independent, and the expected number of throws needed to hit the target region is geometrically distributed with a success probability of $D_{\tv}(\mathcal{P}, \mathcal{Q})$. Therefore, the expected number of throws needed is $\frac{1}{D_{\tv}(\mathcal{P}, \mathcal{Q})}$. Bob's dart throwing procedure is only used if $a$ is rejected at Line 6 of Bob's protocol, which happens with probability $D_{\tv}(\mathcal{P}, \mathcal{Q})$.
So, the expected number of communication rounds between Alice and Bob, denoted by $\mathbb{E}[c]$, is:
\begin{align*}
    \mathbb{E}[c] = \Pr[a = b] \cdot 1 + \Pr[a \neq b] \cdot \frac{1}{D_{\tv}(\mathcal{P}, \mathcal{Q})} \leq 1 + D_{\tv}(\mathcal{P}, \mathcal{Q}) \cdot \frac{1}{D_{\tv}(\mathcal{P}, \mathcal{Q})} \leq 2.&\qedhere
\end{align*}
\end{proof}

With \Cref{claim:exp_communication} in place, we are almost ready to prove \Cref{clm:lowcomm}. The only issue is that Protocol \ref{prot:low_communication} must communicate probabilities, which could be arbitrarily precise real numbers, and thus require an unbounded number of bits to represent. 
To avoid this issue,  we need to discretize $\mathcal{P}$ and $\mathcal{Q}$ using \Cref{alg:rounding}, to form distributions $\mathcal{P}'$ and $\mathcal{Q}'$, where each probability is a multiple of $\frac{\epsilon}{n}$ for a given small $\epsilon$. This quantization allows us to communicate the probabilities using just $O(\log(n/\epsilon))$ bits.

\begin{algorithm}
\caption{Distribution Discretization}
\label{alg:rounding}
\begin{algorithmic}[1]
\Require Original distribution $\mathcal{P} = \{p_1, \ldots, p_n\}$, precision parameter $\epsilon$, and size $n$.
\Ensure Distribution $\mathcal{P}' = \{p_1', \ldots, p_n'\}$ where each $p_i'$ is a multiple of $\epsilon/n$.
\For{$p_i \in \mathcal{P}$}
    \State Set $p_i'$ equal to $p_i$ rounded down to the nearest multiple of $\frac{\epsilon}{n}$.
\EndFor
\State Set $r \gets (1 - \sum_{i=1}^n p_i')$.
\State Set $p_1' \gets p_1' + r$. \Comment{This step ensures that $\sum_{i=1}^n p_i' = 1$.}
\end{algorithmic}
\end{algorithm}

\begin{lemma}
\label{lemma:dtv_rounded}
    For any distributions $\mathcal{P}$ and $\mathcal{Q}$ over $\{1,\ldots, n\}$, if we apply \Cref{alg:rounding} with parameters $\epsilon,n$ to get distributions $\mathcal{P}'$ and $\mathcal{Q}'$, then we have:
    \begin{align*}
        D_{\tv}(\mathcal{P}, \mathcal{P}') &\leq \epsilon,  & D_{\tv}(\mathcal{Q}, \mathcal{Q}') \leq \epsilon, & & \text{and} & & D_{\tv}(\mathcal{P}', \mathcal{Q}') &\leq D_{\tv}(\mathcal{P}, \mathcal{Q}) + 2\epsilon.
    \end{align*}
\end{lemma}

\begin{proof}
First, consider $\mathcal{P}$. Let $r$ be as in Line 3 of \Cref{alg:rounding}. We have:
\begin{align*}
    r = 1 - \sum_{i=1}^n p_i' 
    \leq 1 - \sum_{i=1}^n \left(\frac{p_i}{\frac{\epsilon}{n}} - 1\right) \cdot \frac{\epsilon}{n} 
    = 1 - \sum_{i=1}^n \left(p_i - \frac{\epsilon}{n}\right)
    = \sum_{i=1}^n \frac{\epsilon}{n} = \epsilon.
\end{align*}
We can then bound the total variation distance between $\mathcal{P}$ and $\mathcal{P}'$ as:
\begin{align*}
    D_{\tv}(\mathcal{P}, \mathcal{P}') &=  \frac{1}{2}\sum_{i=1}^n |p_i-p'_i|\leq \frac{1}{2}\cdot 2r \leq \epsilon.
\end{align*}
We can similarly demonstrate that $D_{\tv}(\mathcal{Q}, \mathcal{Q}') < \epsilon$. Finally, the last bound follows from triangle inequality:
\begin{align*}
    D_{\tv}(\mathcal{P}', \mathcal{Q}') 
    &\leq D_{\tv}(\mathcal{P}', \mathcal{P}) + D_{\tv}(\mathcal{P}, \mathcal{Q}) + D_{\tv}(\mathcal{Q}, \mathcal{Q}')
    \leq D_{\tv}(\mathcal{P}, \mathcal{Q}) + 2\epsilon&\qedhere
\end{align*}
\end{proof}

Our final approach will be for Alice and Bob to run 
Protocol \ref{prot:low_communication} on the distributions $\mathcal{P}'$ and $\mathcal{Q}'$ instead of $\mathcal{P}$ and $\mathcal{Q}$. This produces samples $a\sim \mathcal{P}', b\sim \mathcal{Q}'$ instead of $a\sim \mathcal{P}, b\sim \mathcal{Q}$, as required. To correct the distributions of the samples, Alice and Bob simply use an optimal coupling between $\mathcal{P}$ and $\mathcal{P}'$ and between $\mathcal{Q}$ and $\mathcal{Q}'$. These couplings can be computed locally, without any communication. 


\begin{proof}[Proof of \Cref{clm:lowcomm}]
Alice and Bob round their distributions using \Cref{alg:rounding} with parameter $\epsilon/4$ to obtain $\mathcal{P}'$ and $\mathcal{Q}'$.
They then run Protocol \ref{prot:low_communication} with $\mathcal{P}'$ and $\mathcal{Q}'$ to obtain $a'\sim \mathcal{P}'$ and $b'\sim \mathcal{Q}'$ such that $\Pr[a'=b'] \geq 1 - D_{\tv}(\mathcal{P}',\mathcal{Q}')$. They each separately apply Protocol \ref{prot:coupling} to obtain samples $a\sim \mathcal{P}$ and $b\sim \mathcal{Q}$ such that $\Pr[a=a'] \geq 1-D_{\tv}(\mathcal{P},\mathcal{P}')$ and $\Pr[b=b'] \geq 1-D_{\tv}(\mathcal{Q},\mathcal{Q}')$. Overall, we have:
\begin{align*}
    \Pr[a = b] &\geq \Pr[a = a' \text{ and } a' = b' \text{ and } b = b'] \\ 
    &\geq 1 - D_{\tv}(\mathcal{P}, \mathcal{P}') -D_{\tv}(\mathcal{Q}', \mathcal{P}') - D_{\tv}(\mathcal{Q}, \mathcal{Q}') \tag{by a union bound}\\
    &\geq 1- D_{\tv}(\mathcal{P}, \mathcal{Q}) - 4\cdot\epsilon/4 \tag{by \Cref{lemma:dtv_rounded}}\\
    &= 1 - D_{\tv}(\mathcal{P}, \mathcal{Q}) - \epsilon.
\end{align*}
The total communication cost is that of running Protocol \ref{prot:low_communication} with $\mathcal{P}'$ and $\mathcal{Q}'$. Each round of this protocol requires communicating an index, which takes $O(\log n)$ bits, and a probability, which takes $O(\log(n/\epsilon))$ bits (since all probabilities are integer multiples of $\epsilon/4n$).
By \Cref{claim:exp_communication}, the protocol terminates after a constant number of rounds in expectation, so in total, we need $O(\log(n/\epsilon))$ bits of communication, in expectation.
\end{proof}

\section*{Acknowledgement} This work was partially support by NSF Grant $\# 2106888$. Authors thank Mert Cemri, Nived Rajaraman, and Cheuk Ting Li for pointing out related works.

\bibliographystyle{plainnat}
\bibliography{main}

\end{document}

%% file: bib_macros.tex
  \usepackage{nth}
  \usepackage{intcalc}

  \newcommand{\cAAAI}[1]{AAAI\ Conference\ on\ Artificial (AAAI)}

%% file: main.bbl
\begin{thebibliography}{44}
\providecommand{\natexlab}[1]{#1}
\providecommand{\url}[1]{\texttt{#1}}
\expandafter\ifx\csname urlstyle\endcsname\relax
  \providecommand{\doi}[1]{doi: #1}\else
  \providecommand{\doi}{doi: \begingroup \urlstyle{rm}\Url}\fi

\bibitem[Angel and Spinka(2021)]{angel2021pairwiseoptimalcouplingmultiple}
Omer Angel and Yinon Spinka.
\newblock Pairwise optimal coupling of multiple random variables.
\newblock \emph{\arXiv{1903.00632}}, 2021.

\bibitem[Bavarian et~al.(2020)Bavarian, Ghazi, Haramaty, Kamath, Rivest, and Sudan]{BavarianGhaziHaramaty:2020}
Mohammad Bavarian, Badih Ghazi, Elad Haramaty, Pritish Kamath, Ronald~L. Rivest, and Madhu Sudan.
\newblock Optimality of correlated sampling strategies.
\newblock \emph{Theory of Computing}, 16\penalty0 (12):\penalty0 1--18, 2020.

\bibitem[Bessa et~al.(2023)Bessa, Daliri, Freire, Musco, Musco, Santos, and Zhang]{BessaDaliriFreire:2023}
Aline Bessa, Majid Daliri, Juliana Freire, Cameron Musco, Christopher Musco, A\'{e}cio Santos, and Haoxiang Zhang.
\newblock Weighted minwise hashing beats linear sketching for inner product estimation.
\newblock In \emph{\PODS{2023}}, 2023.

\bibitem[Beyer et~al.(2007)Beyer, Haas, Reinwald, Sismanis, and Gemulla]{BeyerHaasReinwald:2007}
Kevin Beyer, Peter~J. Haas, Berthold Reinwald, Yannis Sismanis, and Rainer Gemulla.
\newblock On synopses for distinct-value estimation under multiset operations.
\newblock In \emph{\SIGMOD{2007}}, pages 199--210, 2007.

\bibitem[Broder(1997)]{Broder:1997}
Andrei~Z. Broder.
\newblock On the resemblance and containment of documents.
\newblock In \emph{Proceedings. Compression and Complexity of SEQUENCES}, pages 21--29, 1997.

\bibitem[Broder et~al.(1998)Broder, Charikar, Frieze, and Mitzenmacher]{BroderCharikarFrieze:1998}
Andrei~Z. Broder, Moses Charikar, Alan~M. Frieze, and Michael Mitzenmacher.
\newblock Min-wise independent permutations (extended abstract).
\newblock In \emph{\STOC{1998}}, pages 327--336, 1998.

\bibitem[Chen et~al.(2023)Chen, Borgeaud, Irving, Lespiau, Sifre, and Jumper]{ChenBorgeaudIrving:2023}
Charlie Chen, Sebastian Borgeaud, Geoffrey Irving, Jean-Baptiste Lespiau, Laurent Sifre, and John Jumper.
\newblock Accelerating large language model decoding with speculative sampling.
\newblock \emph{\arXiv{2302.01318}}, 2023.

\bibitem[Christiani(2020)]{Christiani:2020}
Tobias Christiani.
\newblock {DartMinHash}: Fast sketching for weighted sets.
\newblock \emph{\arXiv{2005.11547}}, 2020.

\bibitem[Cohen(1997)]{Cohen:1997}
Edith Cohen.
\newblock Size-estimation framework with applications to transitive closure and reachability.
\newblock \emph{J. Comput. Syst. Sci.}, 55\penalty0 (3):\penalty0 441--453, 1997.

\bibitem[Cohen(2015)]{Cohen:2015}
Edith Cohen.
\newblock Stream sampling for frequency cap statistics.
\newblock In \emph{\KDD{2015}}, pages 159--168, 2015.

\bibitem[Cohen(2023)]{Cohen:2023}
Edith Cohen.
\newblock Sampling big ideas in query optimization.
\newblock In \emph{\PODS{2023}}, 2023.

\bibitem[Cohen and Kaplan(2007)]{CohenKaplan:2007}
Edith Cohen and Haim Kaplan.
\newblock Summarizing data using bottom-k sketches.
\newblock In \emph{\PODC{2007}}, pages 225--234, 2007.

\bibitem[Daliri et~al.(2024{\natexlab{a}})Daliri, Freire, Musco, Santos, and Zhang]{DaliriFreireMusco:2024}
Majid Daliri, Juliana Freire, Christopher Musco, A{\'e}cio Santos, and Haoxiang Zhang.
\newblock Sampling methods for inner product sketching.
\newblock \emph{Proc. VLDB Endow.}, 2024{\natexlab{a}}.

\bibitem[Daliri et~al.(2024{\natexlab{b}})Daliri, Freire, Musco, Santos, and Zhang]{DaliriFreireMusco:2024b}
Majid Daliri, Juliana Freire, Christopher Musco, A{\'e}cio Santos, and Haoxiang Zhang.
\newblock Simple analysis of priority sampling.
\newblock \emph{\SOSA{2024}}, 2024{\natexlab{b}}.

\bibitem[Duffield et~al.(2004)Duffield, Lund, and Thorup]{DuffieldLundThorup:2004}
Nick Duffield, Carsten Lund, and Mikkel Thorup.
\newblock Flow sampling under hard resource constraints.
\newblock In \emph{Proceedings of the Joint International Conference on Measurement and Modeling of Computer Systems (SIGMETRICS)}, pages 85--96, 2004.

\bibitem[Duffield et~al.(2005)Duffield, Lund, and Thorup]{DuffieldLundThorup:2005}
Nick Duffield, Carsten Lund, and Mikkel Thorup.
\newblock Learn more, sample less: control of volume and variance in network measurement.
\newblock \emph{IEEE Transactions on Information Theory}, 51\penalty0 (5):\penalty0 1756--1775, 2005.

\bibitem[Efraimidis and Spirakis(2006)]{EfraimidisSpirakis:2006}
Pavlos~S. Efraimidis and Paul~G. Spirakis.
\newblock Weighted random sampling with a reservoir.
\newblock \emph{Information Processing Letters}, 97\penalty0 (5):\penalty0 181--185, 2006.

\bibitem[Estan and Naughton(2006)]{EstanNaughton:2006}
C.~Estan and J.F. Naughton.
\newblock End-biased samples for join cardinality estimation.
\newblock In \emph{\ICDE{2006}}, 2006.

\bibitem[Flajolet(1990)]{Flajolet:1990}
Philippe Flajolet.
\newblock On adaptive sampling.
\newblock \emph{Computing}, 43\penalty0 (4):\penalty0 391--400, 1990.

\bibitem[{Gemini Team, Google}(2023)]{team2023gemini}
{Gemini Team, Google}.
\newblock Gemini: a family of highly capable multimodal models.
\newblock \emph{\arXiv{2312.11805}}, 2023.

\bibitem[{Gemma Team, Google}(2024)]{gemmateam2024gemma2improvingopen}
{Gemma Team, Google}.
\newblock Gemma 2: Improving open language models at a practical size.
\newblock \emph{\arXiv{2408.00118}}, 2024.

\bibitem[Gumbel(1935)]{gumbel1935valeurs}
Emil~Julius Gumbel.
\newblock Les valeurs extr{\^e}mes des distributions statistiques.
\newblock In \emph{Annales de l'institut Henri Poincar{\'e}}, volume~5, pages 115--158, 1935.

\bibitem[Haeupler et~al.(2014)Haeupler, Manasse, and Talwar]{HaeuplerManasseTalwar:2014}
Bernhard Haeupler, Mark Manasse, and Kunal Talwar.
\newblock Consistent weighted sampling made fast, small, and easy.
\newblock \emph{\arXiv{1410.4266}}, 2014.

\bibitem[Hoffmann et~al.(2022)Hoffmann, Borgeaud, Mensch, Buchatskaya, Cai, Rutherford, de~Las~Casas, Hendricks, Welbl, Clark, Hennigan, Noland, Millican, van~den Driessche, Damoc, Guy, Osindero, Simonyan, Elsen, Vinyals, Rae, and Sifre]{HoffmannBorgeaudMensch:2022}
Jordan Hoffmann, Sebastian Borgeaud, Arthur Mensch, Elena Buchatskaya, Trevor Cai, Eliza Rutherford, Diego de~Las~Casas, Lisa~Anne Hendricks, Johannes Welbl, Aidan Clark, Tom Hennigan, Eric Noland, Katie Millican, George van~den Driessche, Bogdan Damoc, Aurelia Guy, Simon Osindero, Karen Simonyan, Erich Elsen, Oriol Vinyals, Jack~W. Rae, and Laurent Sifre.
\newblock Training compute-optimal large language models.
\newblock In \emph{\NIPS{2022}}, 2022.

\bibitem[Holenstein(2009)]{Holenstein:2009}
Thomas Holenstein.
\newblock Parallel repetition: Simplification and the no-signaling case.
\newblock \emph{Theory of Computing}, 5\penalty0 (8):\penalty0 141--172, 2009.

\bibitem[Huijben et~al.(2023)Huijben, Kool, Paulus, and van Sloun]{HuijbenKoolPaulus:2023}
Iris A.~M. Huijben, Wouter Kool, Max~B. Paulus, and Ruud~G. van Sloun.
\newblock A review of the {G}umbel-max trick and its extensions for discrete stochasticity in machine learning.
\newblock \emph{IEEE Transactions on Pattern Analysis and Machine Intelligence}, 45\penalty0 (02):\penalty0 1353--1371, 2023.

\bibitem[Ioffe(2010)]{Ioffe:2010}
Sergey Ioffe.
\newblock Improved consistent sampling, weighted minhash and l1 sketching.
\newblock In \emph{\ICDM{2010}}, pages 246--255, 2010.

\bibitem[Kleinberg and Tardos(2002)]{KleinbergTardos:2002}
Jon Kleinberg and \'{E}va Tardos.
\newblock Approximation algorithms for classification problems with pairwise relationships: metric labeling and markov random fields.
\newblock \emph{J. ACM}, 49\penalty0 (5):\penalty0 616–639, 2002.

\bibitem[Kool et~al.(2019)Kool, Van~Hoof, and Welling]{kool2019stochastic}
Wouter Kool, Herke Van~Hoof, and Max Welling.
\newblock Stochastic beams and where to find them: The gumbel-top-k trick for sampling sequences without replacement.
\newblock In \emph{\ICML{2019}}, pages 3499--3508. PMLR, 2019.

\bibitem[Leviathan et~al.(2023)Leviathan, Kalman, and Matias]{LeviathanKalmanMatias:2023}
Yaniv Leviathan, Matan Kalman, and Yossi Matias.
\newblock Fast inference from transformers via speculative decoding.
\newblock In \emph{\ICML{2023}}, 2023.

\bibitem[Li and Anantharam(2019)]{li2019pairwisemultimarginaloptimaltransport}
Cheuk~Ting Li and Venkat Anantharam.
\newblock Pairwise multi-marginal optimal transport and embedding for earth mover's distance.
\newblock \emph{\arXiv{1908.01388}}, 2019.

\bibitem[Li and Anantharam(2021)]{li2021unified}
Cheuk~Ting Li and Venkat Anantharam.
\newblock A unified framework for one-shot achievability via the poisson matching lemma.
\newblock \emph{IEEE Transactions on Information Theory}, 67\penalty0 (5):\penalty0 2624--2651, 2021.

\bibitem[Li(2017)]{Li:2017}
Ping Li.
\newblock Linearized {GMM} kernels and normalized random fourier features.
\newblock In \emph{\KDD{2017}}, pages 315--324, 2017.

\bibitem[Li et~al.(2006)Li, Church, and Hastie]{LiChurchHastie:2006}
Ping Li, Kenneth Church, and Trevor Hastie.
\newblock Conditional random sampling: A sketch-based sampling technique for sparse data.
\newblock In \emph{\NIPS{2006}}, volume~19, 2006.

\bibitem[Maddison et~al.(2014)Maddison, Tarlow, and Minka]{MaddisonTarlowMinka:2014}
Chris~J. Maddison, Daniel Tarlow, and Tom Minka.
\newblock {A*} sampling.
\newblock In \emph{\NIPS{2014}}, pages 3086--3094, 2014.

\bibitem[Manasse et~al.(2010)Manasse, McSherry, and Talwar]{ManasseMcSherryTalwar:2010}
Mark Manasse, Frank McSherry, and Kunal Talwar.
\newblock Consistent weighted sampling.
\newblock Technical Report MSR-TR-2010-73, Microsoft Research, 2010.

\bibitem[OpenAI(2023)]{achiam2023gpt}
OpenAI.
\newblock {GPT-4} technical report.
\newblock \emph{\arXiv{2303.08774}}, 2023.

\bibitem[Ros{\'e}n(1997)]{Rosen:1997}
Bengt Ros{\'e}n.
\newblock Asymptotic theory for order sampling.
\newblock \emph{Journal of Statistical Planning and Inference}, 62\penalty0 (2):\penalty0 135--158, 1997.

\bibitem[Shrivastava(2016)]{Shrivastava:2016}
Anshumali Shrivastava.
\newblock Simple and efficient weighted minwise hashing.
\newblock In \emph{\NIPS{2016}}, 2016.

\bibitem[Sun et~al.(2023)Sun, Suresh, Ro, Beirami, Jain, and Yu]{SunSureshRo:2023}
Ziteng Sun, Ananda~Theertha Suresh, Jae~Hun Ro, Ahmad Beirami, Himanshu Jain, and Felix Yu.
\newblock {SpecTr}: Fast speculative decoding via optimal transport.
\newblock In \emph{\NIPS{2023}}, 2023.

\bibitem[Touvron et~al.(2023)Touvron, Lavril, Izacard, Martinet, Lachaux, Lacroix, Rozi{\`e}re, Goyal, Hambro, Azhar, et~al.]{touvron2023llama}
Hugo Touvron, Thibaut Lavril, Gautier Izacard, Xavier Martinet, Marie-Anne Lachaux, Timoth{\'e}e Lacroix, Baptiste Rozi{\`e}re, Naman Goyal, Eric Hambro, Faisal Azhar, et~al.
\newblock Llama: Open and efficient foundation language models.
\newblock \emph{\arXiv{2302.13971}}, 2023.

\bibitem[Wu et~al.(2020)Wu, Li, Chen, Gao, and Zhang]{WuLiChen:2020}
Wei Wu, Bin Li, Ling Chen, Junbin Gao, and Chengqi Zhang.
\newblock A review for weighted {M}in{H}ash algorithms.
\newblock \emph{IEEE Trans. Knowl. Data Eng.}, 2020.

\bibitem[Wu(2020)]{wu2020information}
Yihong Wu.
\newblock Information-theoretic methods for high-dimensional statistics.
\newblock \emph{Lecture notes, Yale University, New Haven, CT}, 2020.

\bibitem[Zhou et~al.(2024)Zhou, Ning, Hong, Fu, Xu, Li, Lou, Wang, Yuan, Li, et~al.]{zhou2024survey}
Zixuan Zhou, Xuefei Ning, Ke~Hong, Tianyu Fu, Jiaming Xu, Shiyao Li, Yuming Lou, Luning Wang, Zhihang Yuan, Xiuhong Li, et~al.
\newblock A survey on efficient inference for large language models.
\newblock \emph{\arXiv{2404.14294}}, 2024.

\end{thebibliography}
